\documentclass{article}

\usepackage{calc}
\usepackage{amsmath}
\usepackage{amsxtra}
\usepackage{amssymb}
\usepackage{amsfonts}
\usepackage{amsthm}
\usepackage{amstext}
\usepackage{amsbsy}
\usepackage{amscd}
\usepackage{color}
\usepackage[sans]{dsfont}

\usepackage{fancyhdr}

\usepackage{graphicx}
\usepackage{texdraw}

\usepackage{multind}

\theoremstyle{plain}
\newtheorem{theorem}{Theorem}
\newtheorem{lemma}{Lemma}

\newtheorem{proposition}[lemma]{Proposition}
\theoremstyle{definition}
\newtheorem{convention}[lemma]{Convention}

\newtheorem{definition}[lemma]{Definition}
\theoremstyle{remark}
\newtheorem{remark}[lemma]{Remark}

\newcommand{\R}{\mathds{R}}
\newcommand{\N}{\mathds{N}}

\newcommand{\defm}[1]{\emph{#1}}
\newcommand{\subeq}[2]{\mathord{\underbrace{\mathop{#1}}_{#2}}}

\newcommand{\sign}{\operatorname{sgn}}

\newcommand{\spC}{\mathcal{C}}

\newcommand{\nva}{\beta}

\newcommand{\contemb}{\hookrightarrow}

\newcommand{\fullB}{\underline{B}}

\newcommand{\xo}{\chi}
\newcommand{\po}{\psi}
\newcommand{\so}{s}

\newcommand{\xn}{\hat\chi}
\newcommand{\pn}{\hat\psi}

\newcommand{\pnl}{\acute\psi}

\newcommand{\etat}{\eta^*}
\newcommand{\vxit}{\vec\xi^*}

\newcommand{\etaa}{\eta}
\newcommand{\vxia}{\vec\xi}

\makeatletter\@addtoreset{equation}{subsection}\@addtoreset{equation}{section}\makeatother
\makeatletter\@addtoreset{lemma}{subsection}\@addtoreset{equation}{section}\makeatother

\newcommand{\myeqref}[1]{(\ref{#1})}
\newcommand{\myref}[1]{\ref{#1}}

\newcommand{\mylabel}[1]{\label{#1}}
\newcommand{\myeqlabel}[1]{\label{#1}}

\newcommand{\sindex}[2]{}

\usepackage{fancyhdr}
\def\pmpaper{elling-liu-pmeyer-arxiv}
\def\pmnewlabel#1#2{\expandafter\def\csname pmref-#1\endcsname{#2}}
\def\pmref#1{\csname pmref-#1\endcsname}
\def\pmeqref#1{(\csname pmref-#1\endcsname)}
\def\pmc#1{\cite[#1]{\pmpaper}}

\pmnewlabel{fig:weakstrong}{1}
\pmnewlabel{fig:wedgeflow}{2}
\pmnewlabel{fig:spolar-lines}{3}
\pmnewlabel{section:numerics}{1.2}
\pmnewlabel{fig:fullsol}{4}
\pmnewlabel{fig:frameorig}{4}
\pmnewlabel{th:elling-liu}{1}
\pmnewlabel{fig:numerics}{5}
\pmnewlabel{eq:techcond}{1.3.1}
\pmnewlabel{eq:prob1}{1.3.2}
\pmnewlabel{eq:prob2}{1.3.3}
\pmnewlabel{eq:prob3}{1.3.4}
\pmnewlabel{rem:weaksol}{1.3.1}
\pmnewlabel{eq:weak-ini}{1.3.5}
\pmnewlabel{fig:techcond}{6}
\pmnewlabel{section:ppf}{2}
\pmnewlabel{section:pf}{2.1}
\pmnewlabel{eq:rhodiv}{2.1.1}
\pmnewlabel{eq:mom}{2.1.2}
\pmnewlabel{eq:p-polytropic}{2.1.3}
\pmnewlabel{eq:v}{2.1.4}
\pmnewlabel{eq:rhoA}{2.1.5}
\pmnewlabel{eq:potflow-divform}{2.1.6}
\pmnewlabel{eq:rho}{2.1.7}
\pmnewlabel{eq:Dpiinv}{2.1.8}
\pmnewlabel{eq:potential-flow}{2.1.9}
\pmnewlabel{eq:cs-uspf}{2.1.10}
\pmnewlabel{section:sspf}{2.2}
\pmnewlabel{eq:psi-phi}{2.2.1}
\pmnewlabel{eq:rhoeq}{2.2.2}
\pmnewlabel{eq:chi-divform}{2.2.3}
\pmnewlabel{eq:chi}{2.2.4}
\pmnewlabel{eq:psi}{2.2.5}
\pmnewlabel{eq:css}{2.2.6}
\pmnewlabel{rem:symmetries}{2.2.1}
\pmnewlabel{eq:L}{2.2.7}
\pmnewlabel{eq:chijump}{2.3.1}
\pmnewlabel{eq:momjump}{2.3.2}
\pmnewlabel{eq:psijump}{2.3.3}
\pmnewlabel{eq:chitan}{2.3.4}
\pmnewlabel{eq:psitan}{2.3.5}
\pmnewlabel{eq:steady-continuity}{2.3.6}
\pmnewlabel{eq:chitan-z}{2.3.7}
\pmnewlabel{section:shocks}{2.4}
\pmnewlabel{eq:steady-rho}{2.4.1}
\pmnewlabel{eq:normal-rho}{2.4.2}
\pmnewlabel{lemma:srel-M}{2.4.1}
\pmnewlabel{eq:srel-M}{2.4.3}
\pmnewlabel{eq:g}{2.4.4}
\pmnewlabel{eq:dgdM}{2.4.5}
\pmnewlabel{eq:ddgddM}{2.4.6}
\pmnewlabel{eq:gM0}{2.4.7}
\pmnewlabel{eq:gMinf}{2.4.8}
\pmnewlabel{eq:cM}{2.4.9}
\pmnewlabel{eq:rhoM}{2.4.10}
\pmnewlabel{eq:cMpre}{2.4.11}
\pmnewlabel{lemma:srel}{2.4.2}
\pmnewlabel{eq:Mnd-asym}{2.4.12}
\pmnewlabel{eq:dMRdMLgen}{2.4.13}
\pmnewlabel{eq:dMRdML}{2.4.14}
\pmnewlabel{prop:shockv}{2.4.4}
\pmnewlabel{eq:DDspecial}{2.4.16}
\pmnewlabel{eq:DvndDvnu}{2.4.17}
\pmnewlabel{eq:zndznucomp}{2.4.18}
\pmnewlabel{lemma:movingnormal}{2.4.6}
\pmnewlabel{eq:DvndDsigma}{2.4.19}
\pmnewlabel{eq:DrhoDsigma}{2.4.20}
\pmnewlabel{section:shockpolar}{2.5}
\pmnewlabel{prop:shockpolar}{2.5.1}
\pmnewlabel{eq:DvdxDnva}{2.5.1}
\pmnewlabel{eq:DvdyDnva}{2.5.2}
\pmnewlabel{eq:blark}{2.5.3}
\pmnewlabel{fig:shockL}{7}
\pmnewlabel{prop:steady-shock-circle}{2.6.1}
\pmnewlabel{prop:vdzero}{2.6.2}
\pmnewlabel{fig:horvzero}{8}
\pmnewlabel{eq:bluaaa}{2.6.1}
\pmnewlabel{eq:vydeta}{2.6.2}
\pmnewlabel{eq:etaMbeta}{2.6.3}
\pmnewlabel{section:apriori}{3}
\pmnewlabel{lemma:alemma2}{3.1.1}
\pmnewlabel{eq:alemma2}{3.1.1}
\pmnewlabel{lemma:alemma}{3.1.2}
\pmnewlabel{eq:alemma}{3.1.2}
\pmnewlabel{eq:aineq}{3.1.3}
\pmnewlabel{lemma:interior-higher}{3.1.3}
\pmnewlabel{eq:lihcond}{3.1.4}
\pmnewlabel{eq:genDpsi}{3.1.5}
\pmnewlabel{eq:genDf}{3.1.6}
\pmnewlabel{prop:c-principle}{3.2.1}
\pmnewlabel{eq:rhomin2}{3.2.1}
\pmnewlabel{prop:interior-velocity}{3.3.1}
\pmnewlabel{prop:v-wall}{3.4.1}
\pmnewlabel{prop:vshock}{3.5.1}
\pmnewlabel{eq:propvx1}{3.5.1}
\pmnewlabel{eq:propvx-s11}{3.5.2}
\pmnewlabel{eq:blurb}{3.5.3}
\pmnewlabel{eq:rhos-psi12}{3.5.4}
\pmnewlabel{eq:wsys}{3.5.5}
\pmnewlabel{eq:wdet}{3.5.6}
\pmnewlabel{prop:L-minmax}{3.6.1}
\pmnewlabel{eq:rhoLb}{3.6.1}
\pmnewlabel{eq:rho-sstr}{3.6.2}
\pmnewlabel{eq:L2}{3.6.4}
\pmnewlabel{prop:density-shock}{3.7.1}
\pmnewlabel{fig:rhos11}{9}
\pmnewlabel{fig:rhotangent}{10}
\pmnewlabel{fig:shockstrength}{11}
\pmnewlabel{section:ellreg}{4}
\pmnewlabel{section:approach}{4.2}
\pmnewlabel{fig:frameR}{12}
\pmnewlabel{fig:frameL}{12}
\pmnewlabel{fig:regularized}{13}
\pmnewlabel{eq:chitt-exp}{4.2.1}
\pmnewlabel{eq:chitt-exp2}{4.2.2}
\pmnewlabel{eq:chitt-exp3}{4.2.3}
\pmnewlabel{section:parmset}{4.3}
\pmnewlabel{eq:deltant}{4.3.1}
\pmnewlabel{eq:dntcond}{4.3.2}
\pmnewlabel{eq:constlist}{4.3.3}
\pmnewlabel{def:Lambda}{4.3.1}
\pmnewlabel{eq:Ceta}{4.3.4}
\pmnewlabel{lemma:etax}{4.3.2}
\pmnewlabel{fig:etaL}{14}
\pmnewlabel{lemma:Gamma-connected}{4.3.3}
\pmnewlabel{eq:ueLbd}{4.3.5}
\pmnewlabel{def:weighted-hoelder}{4.4.1}
\pmnewlabel{fig:onion}{15}
\pmnewlabel{def:b}{4.4.2}
\pmnewlabel{def:fusp}{4.4.3}
\pmnewlabel{eq:Tt}{4.4.1}
\pmnewlabel{eq:regularity}{4.4.2}
\pmnewlabel{eq:sdef}{4.4.3}
\pmnewlabel{eq:shockwall}{4.4.4}
\pmnewlabel{eq:s-welldef}{4.4.5}
\pmnewlabel{eq:cornerregion}{4.4.6}
\pmnewlabel{eq:cornercone}{4.4.7}
\pmnewlabel{eq:rhoprep}{4.4.8}
\pmnewlabel{eq:rhomin}{4.4.9}
\pmnewlabel{eq:ellip}{4.4.10}
\pmnewlabel{eq:ellipC}{4.4.11}
\pmnewlabel{eq:oldnew}{4.4.12}
\pmnewlabel{eq:tilderho}{4.4.13}
\pmnewlabel{eq:faken}{4.4.14}
\pmnewlabel{eq:tildeL}{4.4.15}
\pmnewlabel{eq:itn-inner}{4.4.16}
\pmnewlabel{eq:itn-parabolic}{4.4.17}
\pmnewlabel{eq:itn-shock}{4.4.18}
\pmnewlabel{eq:itn-wall}{4.4.19}
\pmnewlabel{eq:pn}{4.4.20}
\pmnewlabel{eq:lip}{4.4.21}
\pmnewlabel{eq:partan}{4.4.22}
\pmnewlabel{eq:parnor}{4.4.23}
\pmnewlabel{eq:horvel}{4.4.24}
\pmnewlabel{eq:vertvel}{4.4.25}
\pmnewlabel{eq:leftvel}{4.4.26}
\pmnewlabel{eq:shocknormal}{4.4.27}
\pmnewlabel{eq:ndobb}{4.4.28}
\pmnewlabel{eq:Gb}{4.4.29}
\pmnewlabel{rem:fp}{4.4.4}
\pmnewlabel{rem:reflection}{4.4.5}
\pmnewlabel{prop:Lxn-iso}{4.4.6}
\pmnewlabel{eq:eqlin}{4.4.30}
\pmnewlabel{eq:shocklin}{4.4.31}
\pmnewlabel{eq:parlin}{4.4.32}
\pmnewlabel{eq:walllin}{4.4.33}
\pmnewlabel{prop:pn-uqcont}{4.4.7}
\pmnewlabel{prop:fusp-topology}{4.4.8}
\pmnewlabel{def:it}{4.4.9}
\pmnewlabel{rem:mutrans}{4.5.1}
\pmnewlabel{prop:regularity}{4.5.2}
\pmnewlabel{prop:fp-regularity}{4.5.2}
\pmnewlabel{eq:sLip}{4.5.1}
\pmnewlabel{eq:sregu}{4.5.2}
\pmnewlabel{eq:reguint}{4.5.3}
\pmnewlabel{eq:regu2}{4.5.4}
\pmnewlabel{eq:lipode}{4.5.5}
\pmnewlabel{prop:it-continuous-compact}{4.5.3}
\pmnewlabel{section:L-control}{4.6}
\pmnewlabel{prop:Lbounds}{4.6.1}
\pmnewlabel{eq:Leps}{4.6.1}
\pmnewlabel{fig:pararc}{16}
\pmnewlabel{section:parcs}{4.7}
\pmnewlabel{section:c-pararc}{4.7}
\pmnewlabel{eq:Lsimple}{4.7.1}
\pmnewlabel{eq:Ltchi}{4.7.2}
\pmnewlabel{eq:Lnchi}{4.7.3}
\pmnewlabel{eq:psitautau}{4.7.5}
\pmnewlabel{eq:refp}{4.7.6}
\pmnewlabel{eq:cc-arc}{4.7.7}
\pmnewlabel{eq:h0}{4.7.8}
\pmnewlabel{eq:refk}{4.7.9}
\pmnewlabel{eq:pphi}{4.7.10}
\pmnewlabel{eq:kphi}{4.7.11}
\pmnewlabel{eq:qphi}{4.7.12}
\pmnewlabel{eq:thetaphi}{4.7.13}
\pmnewlabel{prop:pararc}{4.8.1}
\pmnewlabel{eq:rhoP}{4.8.1}
\pmnewlabel{eq:vP}{4.8.2}
\pmnewlabel{eq:thetasector}{4.8.3}
\pmnewlabel{eq:pcontrol}{4.8.4}
\pmnewlabel{eq:ccontrol}{4.8.5}
\pmnewlabel{eq:qcontrol}{4.8.6}
\pmnewlabel{eq:pphi-isen}{4.8.7}
\pmnewlabel{eq:plower}{4.8.8}
\pmnewlabel{eq:pupper}{4.8.9}
\pmnewlabel{fig:thetasector}{17}
\pmnewlabel{section:cornersmoving}{4.9}
\pmnewlabel{eq:zdzuLone}{4.9.1}
\pmnewlabel{eq:zdeta}{4.9.2}
\pmnewlabel{eq:cp1}{4.9.3}
\pmnewlabel{eq:cp2}{4.9.4}
\pmnewlabel{eq:cp3}{4.9.5}
\pmnewlabel{eq:zydpre}{4.9.6}
\pmnewlabel{eq:vyd-eta}{4.9.7}
\pmnewlabel{eq:vyd-eta-positive}{4.9.8}
\pmnewlabel{eq:zxdpre}{4.9.9}
\pmnewlabel{eq:pc-zphi}{4.9.10}
\pmnewlabel{eq:petaineq}{4.9.11}
\pmnewlabel{eq:pc-csq}{4.9.12}
\pmnewlabel{eq:keta}{4.9.13}
\pmnewlabel{eq:pplus}{4.9.14}
\pmnewlabel{eq:kplus}{4.9.15}
\pmnewlabel{eq:qetaB}{4.9.16}
\pmnewlabel{eq:qqq}{4.9.17}
\pmnewlabel{eq:qqeta}{4.9.18}
\pmnewlabel{eq:qeta}{4.9.19}
\pmnewlabel{eq:theta-phibar-plus}{4.9.20}
\pmnewlabel{eq:theta-phibar-minus}{4.9.21}
\pmnewlabel{section:lowerbounds}{4.10}
\pmnewlabel{prop:etaa-lowerbound}{4.10.1}
\pmnewlabel{fig:cmax}{18}
\pmnewlabel{prop:cbar}{4.10.2}
\pmnewlabel{eq:cbar}{4.10.1}
\pmnewlabel{prop:psi-axi}{4.10.3}
\pmnewlabel{eq:a}{4.10.3}
\pmnewlabel{eq:aexp}{4.10.4}
\pmnewlabel{prop:vyd-crit}{4.10.5}
\pmnewlabel{prop:etaa-upperbound}{4.10.6}
\pmnewlabel{section:densitycontrol}{4.11}
\pmnewlabel{prop:rho}{4.11.1}
\pmnewlabel{fig:shocktanarg}{19}
\pmnewlabel{section:v-control}{4.12}
\pmnewlabel{prop:ny}{4.12.1}
\pmnewlabel{prop:vx}{4.12.1}
\pmnewlabel{prop:vy}{4.12.1}
\pmnewlabel{eq:chitshock}{4.12.1}
\pmnewlabel{eq:s1x0}{4.12.2}
\pmnewlabel{prop:oblique-corner}{4.13.1}
\pmnewlabel{eq:gpS}{4.13.2}
\pmnewlabel{prop:fp-boundary}{4.13.2}
\pmnewlabel{section:ls}{4.14}
\pmnewlabel{prop:unperturbed-unique}{4.14.1}
\pmnewlabel{fig:unperturbed}{20}
\pmnewlabel{prop:unperturbed-index}{4.14.3}
\pmnewlabel{eq:paralin}{4.14.1}
\pmnewlabel{eq:fdeta}{4.14.2}
\pmnewlabel{eq:fdnor}{4.14.3}
\pmnewlabel{eq:shockl}{4.14.4}
\pmnewlabel{prop:probell}{4.15.1}
\pmnewlabel{section:entireflow}{4.16}
\pmnewlabel{fig:epslimit}{21}
\pmnewlabel{eq:interior-eps}{4.16.1}
\pmnewlabel{eq:para-chi-eps}{4.16.2}
\pmnewlabel{eq:para-rho-eps}{4.16.3}
\pmnewlabel{eq:para-nablachi-eps}{4.16.4}
\pmnewlabel{eq:shock1-eps}{4.16.5}
\pmnewlabel{eq:shock2-eps}{4.16.6}
\pmnewlabel{eq:cornerdist-eps}{4.16.8}
\pmnewlabel{eq:lip-eps}{4.16.9}
\pmnewlabel{eq:cka-eps}{4.16.10}
\pmnewlabel{eq:eps-weak}{4.16.11}
\pmnewlabel{eq:zero-weak}{4.16.12}
\pmnewlabel{eq:zero-cont}{4.16.13}
\pmnewlabel{eq:zero-cont-rhov}{4.16.14}
\pmnewlabel{rem:notvoid}{4.16.1}
\pmnewlabel{rem:structure}{4.16.2}
\pmnewlabel{fig:corner}{22}
\pmnewlabel{section:corner}{5.1}
\pmnewlabel{prop:corner}{5.1.1}
\pmnewlabel{eq:Gammaregu}{5.1.1}
\pmnewlabel{eq:phithetabd}{5.1.2}
\pmnewlabel{eq:uLip}{5.1.3}
\pmnewlabel{eq:u}{5.1.4}
\pmnewlabel{eq:Gammacond}{5.1.5}
\pmnewlabel{eq:coeffnorm}{5.1.6}
\pmnewlabel{eq:Ce}{5.1.7}
\pmnewlabel{eq:ndob}{5.1.8}
\pmnewlabel{eq:G}{5.1.9}
\pmnewlabel{eq:Colambda}{5.1.10}
\pmnewlabel{eq:ghregu}{5.1.11}
\pmnewlabel{eq:Ctabeta}{5.1.12}
\pmnewlabel{eq:w-interior}{5.1.13}
\pmnewlabel{eq:wboundary}{5.1.14}
\pmnewlabel{eq:ddudg}{5.1.15}
\pmnewlabel{eq:Av}{5.1.16}
\pmnewlabel{eq:gGammaf}{5.1.17}
\pmnewlabel{prop:ztrans}{5.2.1}
\pmnewlabel{eq:Dcond}{5.2.1}
\pmnewlabel{rem:shmor}{5.2.2}

\begin{document}

\title{Regular reflection in self-similar potential flow and the sonic criterion}
\author{Volker Elling}
\date{}	

\maketitle

\begin{abstract}
	Reflection of a shock from a solid wedge is a classical problem in gas dynamics.
	Depending on the parameters either a regular or a irregular (Mach-type) reflection results.
	We construct regular reflection as an exact self-similar solution for potential flow. 
	For some upstream Mach numbers $M_I$ and isentropic coefficients $\gamma$, a solution
	exists for all wedge angles $\theta$ allowed by the \defm{sonic criterion}. This demonstrates that,
	at least for potential flow, weaker criteria are false.
\end{abstract}

\parindent=0cm%
\parskip=\baselineskip%

\section{Introduction}

\subsection{The reflection problem}

\mylabel{section:refl}

Reflection of an incident shock from a solid wedge is a classical problem of gas dynamics. It has been studied extensively by 
Ernst Mach \cite{mach-wosyka,krehl-geest} and John von Neumann \cite{neumann-1943}, as well as many other engineers and mathematicians. 

Most commonly, reflection is studied in \emph{steady} inviscid compressible flow, for example when shocks in a nozzle are reflected from the walls.
The reflections can be classified roughly into \defm{regular} and \defm{irregular reflections};
see \cite{ben-dor-book} for a more detailed discussion. In either type, an \defm{incident shock} $Q$ impinges on a solid surface $\fullB$ 
(see Figure \myref{fig:locrrsmr}).
In regular reflection (RR), $Q$ reaches a \defm{reflection point} on the surface, continuing as the \defm{reflected shock} $R$ (see Figure \myref{fig:locrrsmr} left).

In \defm{irregular reflections} (IRR), incident and reflected shock are connected by a more or less complex interaction pattern which in turn connects
to the solid surface by a third shock, called \defm{Mach stem}.
The most important irregular reflections are double, complex and single Mach reflection (DMR, CMR, SMR); various additional types have been 
proposed \cite{guderley,hunter-brio,hunter-tesdall}. Figure \myref{fig:locrrsmr} right shows an (oversimplified) version of single Mach reflection.

\begin{figure}
	\input{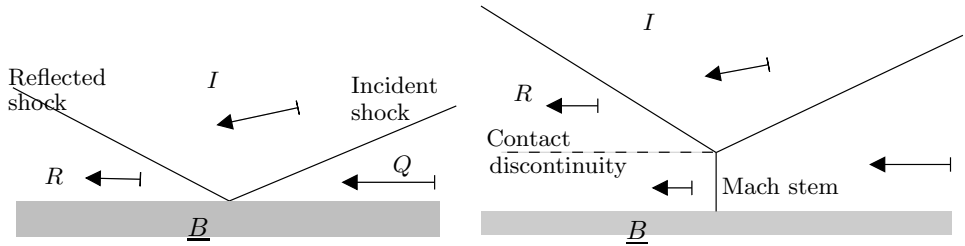}
	\caption{Left: regular reflection, right: single Mach reflection (oversimplified).
	If the reflection point is steady, the sonic criterion corresponds to $M_Q>1$.}
	\mylabel{fig:locrrsmr}
\end{figure}

The reflection problem has several parameters. For polytropic gas it is sufficient to consider the isentropic coefficient $\gamma$ as well as
$L_Q$ and $L_I$, the Mach numbers in the $Q$ 
resp.\ $I$ regions. The incident shock cannot exist unless $L_Q>1$. $L_Q$ and $L_I<L_Q$ determine the incident shock (not all $L_I$ may admit
a matching reflected shock).

In Mach reflection, the Mach stem, reflected and incident shock appear to meet in a \defm{triple point}. In general this is possible
only if they are joined by a contact discontinuity (slip line); for some parameter values it is not possible at all.
In fact for certain values RR is not possible either. This is called the \defm{von Neumann paradox}; it is perhaps the most
famous of the many problems arising in reflection. Many ideas have been proposed towards the resolution of the paradox
(see e.g.\ \cite{guderley,rosales-tabak,gamba-rosales-tabak,hunter-brio,hunter-tesdall}); no single explanation has been accepted widely so far. 

However, this article is concerned with a different question: 
it is natural to ask which parameters cause a RR and which yield IRR. 
Of course both sides of Figure \ref{fig:locrrsmr} are perfectly valid stationary solutions, so the question has to be phrased more carefully.
For example:
\begin{enumerate}
\item Which of the two is dynamically stable (e.g.\ asymptotically stable as a stationary solution of the time-dependent problem)?
\item Which of the two is structurally stable under perturbations like downstream nozzles, wall curvature or 
roughness, interaction with other flow patterns, perturbation of the upstream flow to non-constant with curved incident shock, viscosity, 
heat conduction, boundary layers, noise, slow relaxation to thermal equilibrum and other kinetic effects, dissociation etc.
\end{enumerate}
It is not clear whether these questions are really any better than the original one ---
perhaps both sides of Figure \ref{fig:locrrsmr} are stable.
If so, then the new questions would merely fail in a less obvious way, as stability is harder to check than existence.
But let us assume for the sake of the argument that the vague problem ``does RR or IRR occur'' can be expressed
in some way as a precise mathematical question that selects exactly one of the two choices.

Among the criteria that have been proposed (see \cite[Section 1.5]{ben-dor-book}), three are most important.
The first criterion, called \defm{detachment criterion}, states that RR occurs whenever a reflected shock exists.
Clearly RR is not possible without a reflected shock, so this is the weakest possible criterion.

The velocity $\vec v_I$ in the $I$ region of Figure \myref{fig:rrsmr} forms an angle $\tau$ with $\fullB$; the reflected shock must turn this
velocity by $\tau$ so that $\vec v_R$ is parallel to the wall, satisfying a slip boundary condition.

\begin{figure}
\input{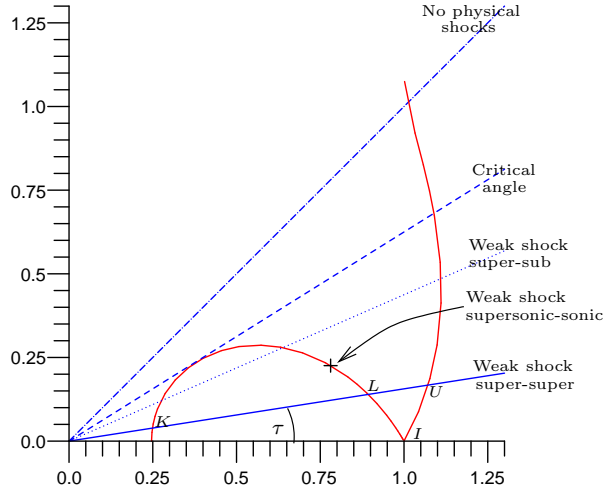}
\caption{A given upstream velocity ($I$) with possible downstream velocities (red curve) 
for steady shocks with varying normal.
The $U$ shock is unphysical; $K$ is the strong and $L$ the weak shock. Shocks cannot turn velocities by more than the critical
angle $\tau_*$.
}
\mylabel{fig:spolar}
\end{figure}

Given the $I$ region data and $\gamma$, let the reflected shock be steady and pass through the reflection point, but vary its angle.
This yields a one-parameter family of velocities $\vec v_R$, forming a curve called \defm{shock polar} 
(see Figure \myref{fig:spolar}). For \emph{physical} shocks there is a maximum angle $\tau_*$ between downstream 
and upstream velocity. $\tau_*$ is determined by the upstream state.

If the angle $\tau$ between wall and $\vec v_I$ region of Figure \myref{fig:rrsmr} right is bigger than
$\tau_*$, no reflected shock exists. If $\tau=\tau_*$, there is exactly one reflected shock. For $\tau<\tau_*$ however there
are \emph{two}, called \defm{weak reflection} and \defm{strong reflection}. We encounter another one of the major issues in reflection:
which of these two should occur? \cite{elling-liu-pmeyer-arxiv} have discussed this question in a related problem.

The flow in the $R$ region can be supersonic or subsonic. If it is supersonic, then waves in the $R$ region cannot travel towards
the reflection point. If it is subsonic, however, they can reach it and interact with it, potentially altering the reflection type.
This motivates the second criterion, called \defm{sonic criterion}: RR occurs exactly if there is a reflected shock with 
supersonic $R$ region, i.e.\ Mach number $L_R>1$.

On the shock polar (Figure \myref{fig:spolar}), $+$ indicates the point where $M_R=1$; 
velocities right of it are supersonic, left of it subsonic. Hence there is an angle $\tau_+$ so that for $\tau<\tau_+$ 
the weak reflection $L$ has $L_R>1$. For $\tau>\tau_+$ however it has $L_R<1$. The strong reflection $K$ is \emph{always} subsonic in the $R$
region --- so the sonic criterion has a pleasant property: only the weak reflection is allowed, solving the uniqueness problem. 
Moreover since $\tau_+<\tau_*$, the sonic criterion is stronger than the detachment criterion.

The third criterion is motivated by studying what happens when the parameters $L_I,L_Q$ are varied so that a transition from RR to IRR occurs.
One might suspect that the pressure in the reflection point in the $R,S$ regions is continuous
and does not jump during transition. Then the pressure behind the reflected shock in RR and the pressure behind the Mach stem in IRR, 
a shock approximately straight and perpendicular to the wall, must be equal at transition. There is a very limited set of $L_I,L_Q,\gamma$ 
for which this happens; 
the \defm{von Neumann criterion} (sometimes called \defm{mechanical equilibrum criterion}) states that the transition can occur only at those parameters.

The von Neumann criterion has various problems. Most importantly, for weak incident shocks the pressure behind the Mach stem 
never matches the pressure below the
reflected shock, so RR should occur in all cases, contradicting observations.

\subsection{Self-similar reflection}

\mylabel{section:refl-selfsim}

Reflection can also be studied in \defm{self-similar} (sometimes called \defm{quasi-steady} or \defm{pseudo-steady}) flow. 
In fact this is advantageous: for finding stationary solutions, choosing boundary conditions that yield well-posedness, in particular
uniqueness, can be rather subtle, as evident from the awkward phrasing of the RR-or-IRR question above.
For initial-value problems, on the other hand, uniqueness
is expected\footnote{\cite{elling-nuq-journal,elling-hyp2004}
raise doubt about the Cauchy problem for the Euler equations, but at least for potential flow the author expects uniqueness to hold.}
 --- or at least a necessary property of any interesting model equation.
Moreover, self-similar flow patterns occur naturally in various reflection experiments.

\begin{figure}
\input{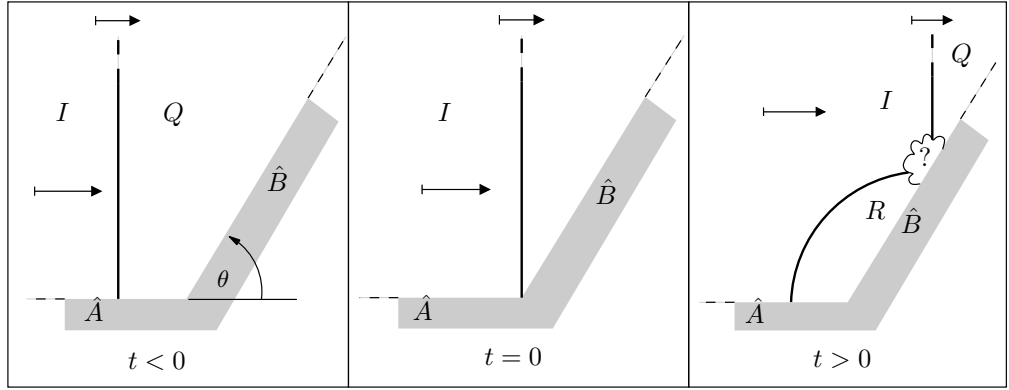}
\caption{Self-similar reflection of a straight vertical shock in a convex corner. Different ``?'' patterns occur depending on corner angle and other parameters.}
\mylabel{fig:rrefini}
\end{figure}
\begin{figure}
	\input{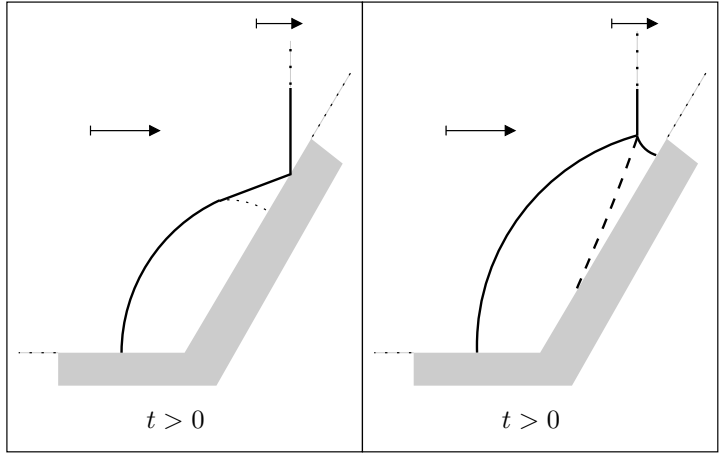}
	\caption{Left: regular reflection. The dotted arc separates a region of constant velocity (above) from a nontrivial region.
	Self-similar potential flow changes type from hyperbolic (above) to parabolic to elliptic across the arc.
	Right: single Mach reflection.}
	\mylabel{fig:rrsmr}
\end{figure}

In self-similar flow, density and velocity are functions of $\xi=x/t$ and $\eta=y/t$ rather than $x,y$. 
To produce a reflection, we consider the horizontal \defm{upstream wall} $\hat A$ and the \defm{downstream wall} $\hat B$ 
(see Figure \myref{fig:rrefini}), meeting in the origin and enclosing an angle $180^\circ-\theta$.
For $t<0$ a vertical incident shock approaches the corner from the left, reaching it at $t=0$; 
for $t>0$ it continues along $\hat B$, while a complex pattern is reflected back from the corner.
For regular reflection, the incident and reflected shock meet in a point $\vec\xi$. 
An observer travelling in the reflection point
will observe a flow expanding at a constant rate, approaching a local RR
as in Figure \myref{fig:locrrsmr} left as $t\uparrow+\infty$. 

To understand self-similarity intuitively, focus on the corner between the two walls in Figure \ref{fig:rrefini} right.
$t\uparrow\infty$ corresponds to zooming into the corner whereas $t\downarrow 0$ corresponds to zooming infinitely far away from the corner.

\begin{figure}
\includegraphics[angle=-90,width=.49\linewidth]{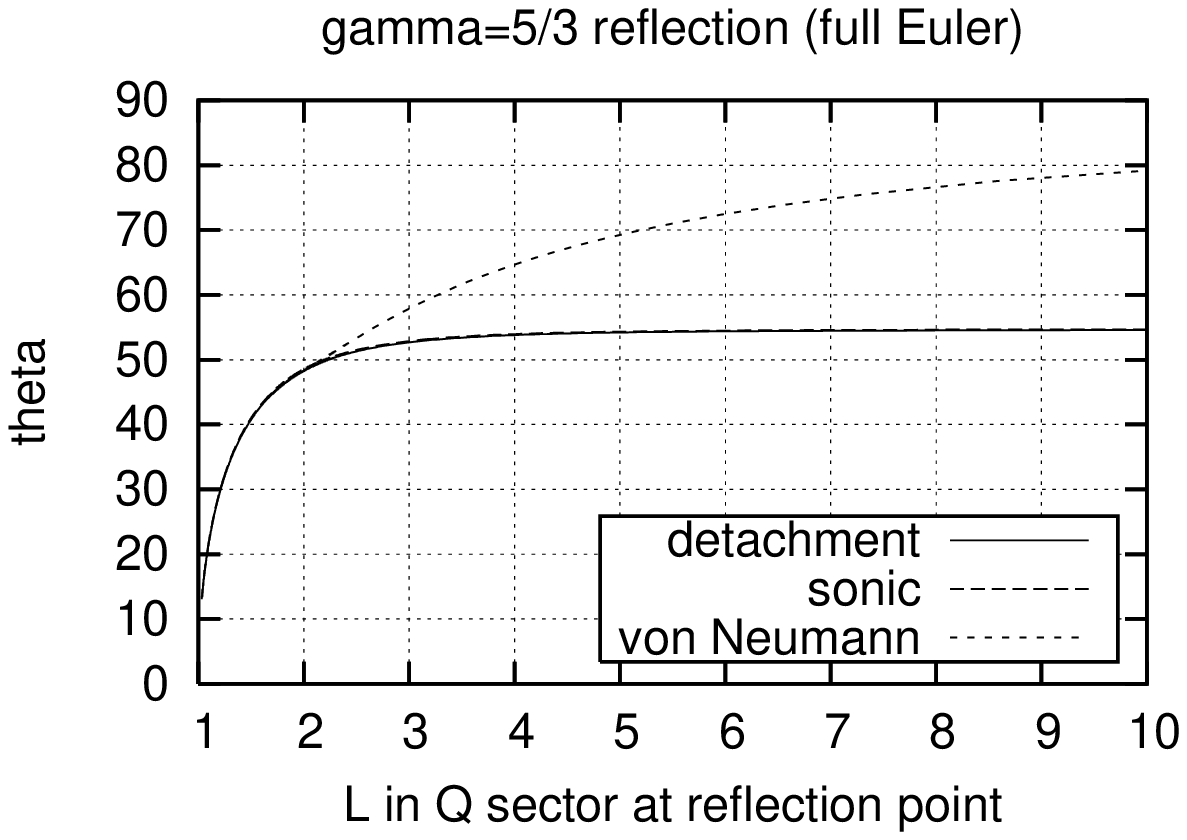}
\includegraphics[angle=-90,width=.49\linewidth]{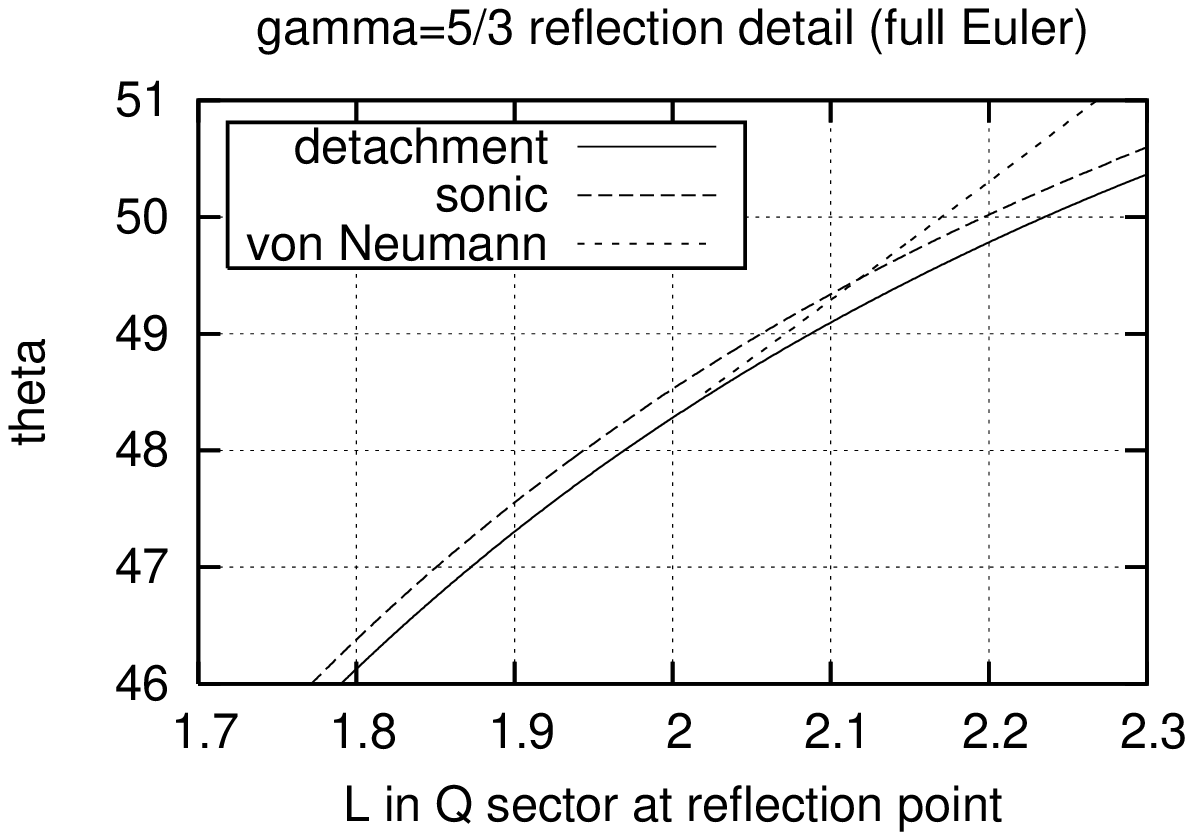}
\caption{Left: transition angles predicted by each criterion (sonic and detachment almost coincide); right: detail.}
\mylabel{fig:transition}
\end{figure}

The three transition criteria discussed for steady reflection specify angles $\theta_d$ (detachment), $\theta_s$ (sonic) and $\theta_N$ (von Neumann), 
depending on $\gamma$ and $L_Q$, so that RR occurs for larger $\theta$ whereas IRR occurs for smaller $\theta$. 
(Here, $L_Q$ is the $Q$ region Mach number as seen by an observer traveling in the intersection point of incident shock and $\hat B$
(= reflection point, in the RR case); of course an observer stationary in the corner will perceive a different velocity in the $Q$ region.)
Note that $\theta_d\leq\theta_s,\theta_N$ always.
Figure \myref{fig:transition} compares the criteria in the case of monatomic gas ($\gamma=5/3$). 

It has also been proposed that the correct criterion may not be the same in steady and self-similar flow (see below),
or that there may be bistable cases where RR and IRR can both occur (see \cite{hornung-oertel-sandeman,ivanov-etal-2001}). 

Nevertheless, it seems that there is an overall preference for the sonic criterion in the scientific community, at least for self-similar reflection.

Numerical and physical experiments are hampered by various difficulties and have not been able to select the correct criterion. 
For example numerical dissipation or physical viscosity smear the shocks and cause boundary layers that interact with the reflection pattern
and can cause ``spurious Mach stems'' \cite{colella-woodward-review}. Moreover, $\theta_d$ and $\theta_s$ are only fractions of a degree
apart (see Figure \myref{fig:transition} right), a resolution that even sophisticated experiments (e.g.\ \cite{lock-dewey}) have been unable to reach.
To quote \cite{ben-dor-book}: ``For this reason it is almost impossible to distinguish experimentally between the sonic and detachment
criteria.''

Constructing exact solutions of most genuinely multi-dimensional flow 
problems is infeasible or restricted to severely simplified equations.
Moreover it would be prohibitively expensive if it could only confirm results that have already been obtained 
many orders of magnitude faster by numerical or physical experiments, unless the certainty of mathematical proof is needed.
Regular reflection appears to be the first instance where rigorous analysis might make a genuine contribution by 
answering a problem that could not be resolved unambiguously by other techniques.

\subsection{Results}

In this article, using techniques developed in \cite{elling-liu-pmeyer-arxiv}, 
regular reflection is constructed as a self-similar solution of compressible potential flow, with
polytropic ($\gamma$-law) gas. 
While classical regular/Mach reflection studies vertical incident shocks, we consider the non-vertical
cases too (these may not arise from any $t<0$ flow), including cases where $\theta>\frac\pi2$.

Most importantly, 
for some values of $\gamma$ and upstream Mach number $M_I$, in particular $\gamma=5/3$ and $M_I=1$,
\emph{every} $\theta$ near $\theta_s$ can be covered. 
This shows rigorously that criteria \emph{stronger} than the sonic criterion are false, at least for potential flow with
this choice of parameters.

As discussed above, there is some tendency to believe that regular reflection does not persist beyond the sonic criterion;
ongoing work aims to show this rigorously, at least under mild assumptions.
This would rule out the \emph{weaker} criteria as well, in particular the detachment criterion, hence prove that sonic is correct.
The problem of weak vs.\ strong reflection (see above) would vanish as well.

However, for now the success is qualified: potential flow lacks contact discontinuities, so \emph{after} the transition to (say) SMR
the flow pattern must be \emph{qualitatively} different from the full Euler flow. It is still possible that the two models 
may have different transition criteria (however, the author believes that this is not the case). 

Although some genuinely multi-dimensional exact solutions have been constructed for steady Euler
flow, self-similar Euler flow is an open and inherently rather difficult problem. But again, it seems unlikely that numerical or experimental 
techniques will yield a clear --- let alone universally accepted --- answer soon, so rigorous analysis would be very valuable.

Here is the precise result:

\begin{figure}
\parbox{.5\textwidth}{\input{theorem.pstex_t}}
\parbox{.5\textwidth}{\input{theorem3.pstex_t}}
\caption{Left: a local RR pattern; right: the curved portion $S$ of the reflected shock has $L_d\leq 1$, hence must be left of the envelope $E$, which bounds it away from the dotted circle and from $\hat B$.}
\mylabel{fig:theorem}
\vskip5mm
\input{theorem2.pstex_t}
\caption{The initial data is constant in each of two sectors that are separated by the incident shock $Q$}
\mylabel{fig:theorem2}
\end{figure}

\begin{theorem}
	\mylabel{th:elling-rrefl}%
	Consider potential flow, as discussed in Section \ref{section:potf}.
	Consider a wall $\hat A=(-\infty,0)\times\{0\}$ (see Figure \myref{fig:theorem}), 
	a second wall ray $\hat B$ at a clockwise angle $180^\circ-\theta$ from $\hat A$, 
	and an incident shock $Q$, at a clockwise angle $180^\circ-\beta_Q$ from $\hat A$, meeting $\hat B$ in the
	reflection point $\vec\xi_R$. Assume that there is a corresponding reflected shock $R$ in $\vec\xi_R$, emanating down and left (or vertically down).
	Define 
	\begin{alignat}{1}
		V_I &:= \{(x,y)\in\R^2:y>0,~-\infty<x<y\cot(\beta_Q)\}\notag \\
		V_Q &:= \{(x,y)\in\R^2:y>0,~y\cot(\beta_Q)<x<y\cot\theta\}\notag, \\
		V &:= \{(x,y)\in\R^2:y>0,~-\infty<x<y\cot\theta\}\notag
	\end{alignat}
	(see Figure \myref{fig:theorem2}).
	\begin{enumerate}
	\item 
		Assume the \defm{sonic criterion} holds: $L_R>1$ in $\vec\xi_R$ in the sector below $R$.
	\item 
		Assume that 
		\begin{alignat}{1}
			|\vec v_I\cdot\vec n_B| &\leq c_I  \mylabel{eq:MyI-one}
		\end{alignat}
	\item 
		Envelope condition:
		of the two\footnote{see Section \myref{section:potf-shocks}} points on the $R$ shock with $L_d=1$, 
		let $\vec\xi^{(0)}_C$ be the one closer to $\vec\xi_R$. 
		Consider shocks with upstream data $\vec v_I,\rho_I$ that go from $\vec\xi^{(0)}_C$ counterclockwise and
		satisfy $L_d\leq 1$ in every point.
		Assume that all such shocks reach $\hat A$ before meeting $\hat B$ or the circle with center $\vec v_I$ and radius $c_I$.
	\end{enumerate}

	Then there exists a weak\footnote{see Remark \myref{rem:weaksol}} solution $\phi=\phi(t,x,y)\in C^{0,1}([0,\infty)\times\overline V)$ of 
	\begin{alignat}{3}
		& \text{unsteady potential flow} \qquad && \text{for $t>0$, $\vec x\in V$,} \myeqlabel{eq:prob1} \\
		& \nabla\phi\cdot\vec n = 0 \qquad && \text{on $\partial V$,} \myeqlabel{eq:prob2} \\
		& \rho = \rho_I, \quad \nabla\phi = \vec v_I \qquad && \text{for $t=0$, $\vec x\in V_I$,} \myeqlabel{eq:prob3} \\
		& \rho = \rho_Q, \quad \nabla\phi = \vec v_Q \qquad && \text{for $t=0$, $\vec x\in V_Q$.} \myeqlabel{eq:prob4} 
	\end{alignat}
\end{theorem}

Of course existence by itself merely validates that potential flow has interesting solutions. In addition, 
detailed results about the structure of the weak solution can be obtained (see Remark \myref{rem:structure}); most importantly,
the flow patterns are of RR type.

\begin{remark}
	\mylabel{rem:weaksol}%
	By weak solution we mean that 
	\begin{alignat}{1}
		\nabla\phi(0,\vec x) &= \vec v_I \qquad\text{for a.e.\ $\vec x\in V_I$} \myeqlabel{eq:weak-ini-1} \\
		\nabla\phi(0,\vec x) &= \vec v_Q \qquad\text{for a.e.\ $\vec x\in V_Q$} \myeqlabel{eq:weak-ini-2}
	\end{alignat}
	and
	\begin{alignat}{1}
		\int_\Omega
		\rho\vartheta_t
		+\rho\nabla\phi\cdot\nabla\vartheta~d\vec x~dt
		+\int_{V_I}\vartheta(0,\vec x)\rho_Id\vec x 
		+\int_{V_Q}\vartheta(0,\vec x)\rho_Qd\vec x 
		&= 0  \notag
	\end{alignat}
	for all test functions $\vartheta\in C_c^\infty(\overline\Omega)$.

	(For $\phi\in C^{0,1}(\overline\Omega)$, the velocity $\nabla\phi$ is a.e.\ well-defined on $\{0\}\times V$, but
	$\phi_t$ and hence $\rho$ may not be well-defined.)
\end{remark}

\begin{remark}
%
	Condition \myeqref{eq:MyI-one} and the envelope condition are merely technical.
	The envelope condition is needed in some cases to prove the shock does not vanish (which is never observed in numerics);
	none of the other estimates requires it.
	Both conditions can probably be removed by future research.
\end{remark}

\subsection{Related work on constructing exact solutions}

In recent years multi-dimensional compressible inviscid flow has received renewed attention, after several recent breakthroughs brought 
the theory of one-dimensional compressible flow to a satisfactory state \cite{glimm,bianchini-bressan-j,liu-yang,bressan-crasta-piccoli}.

\cite{elling-liu-pmeyer-arxiv} (see also \cite{elling-liu-rims05,elling-hyp2006}) studies supersonic flow onto a solid wedge.
For sufficiently sharp wedges, the steady solution consists of a straight shock on each side of the wedge, emanating downstream
and separating two constant-state regions. In inviscid models this shock wave must keep the downstream velocity tangential to the
wedge surface (slip condition). As for regular reflection, there are two different shocks for each (small) wedge angle, a \defm{weak} and a 
\defm{strong} shock. The weak shock is more commonly observed, but no mathematical argument was known to favor it prior to \cite{elling-liu-pmeyer-arxiv}.
In that article, an exact solution was constructed for a wedge at rest in stagnant air, accelerated instantaneously to (sufficiently high) supersonic speed
at time $0$. The resulting flow pattern is self-similar and has a \defm{weak} shock at the wedge tip.

Many of the techniques in \cite{elling-liu-pmeyer-arxiv} are essential in the present article.

The most closely related work, 
and so far the only other paper that proves global existence of some nontrivial time-dependent solution of potential flow is
\cite{chen-feldman-selfsim-journal}: using different techniques, they construct exact solutions for regular reflection, assuming sufficiently 
blunt wedges ($\theta\approx\frac\pi2$). 

Some prior work studies reflection and other problems for simplified models of gas dynamics.
\cite{canic-keyfitz-kim} consider regular reflection for the unsteady transonic small disturbance equation as model.
\cite{yuxi-zheng-rref} studies the same problem for the pressure-gradient system.
The monographs \cite{yuxi-zheng-book,li-zhang-yang} compute various self-similar flows numerically and present some analysis 
and simplified models.

\subsection{Potential flow}

\mylabel{section:potf}

Here we briefly present derivation and elementary results for potential flow. More information can be found in \cite{elling-liu-pmeyer-arxiv}.

Consider the isentropic Euler equations of compressible gas dynamics in $d$ space dimensions:
\begin{alignat}{1}
    \rho_t + \nabla\cdot(\rho\vec v) &= 0 \myeqlabel{eq:rhodiv} \\
    (\rho\vec v)_t + \sum_{i=1}^d(\rho v^i\vec v)_{x^i} + \nabla(p(\rho)) &= 0, \myeqlabel{eq:mom}
\end{alignat}
Hereafter, $\nabla$ denotes the gradient with respect either to the space coordinates ${\vec x}=(x^1,x^2,\cdots,x^d)$ or
the similarity coordinates $t^{-1}\vec x$.
${\vec v}=(v^1,v^2,\cdots,v^d)$ is the velocity of the gas, $\rho$ the density, $p(\rho)$ pressure.
In this article we consider only polytropic pressure laws ($\gamma$-laws) with $\gamma\geq 1$:
\begin{alignat}{1}
    p(\rho) &= \frac{c_0^2\rho_0}{\gamma}\left(\frac{\rho}{\rho_0}\right)^\gamma \myeqlabel{eq:p-polytropic}
\end{alignat}
(here $c_0$ is the sound speed at density $\rho_0$).

For smooth solutions, substituting \myeqref{eq:rhodiv} into \myeqref{eq:mom} yields the simpler form
\begin{alignat}{1}
    \vec v_t+\vec v\cdot\nabla^T\vec v + \nabla(\pi(\rho)) &= 0. \myeqlabel{eq:v}
\end{alignat}
Here $\pi$ is defined as
\begin{alignat}{1}
    \pi(\rho) &= c_0^2\cdot\begin{cases}
	\frac{(\rho/\rho_0)^{\gamma-1}-1}{\gamma-1}, & \gamma> 1 \\
	\log(\rho/\rho_0), & \gamma=1.
    \end{cases}\notag
\end{alignat}
This $\pi$ is $C^\infty$ in $\rho\in(0,\infty)$ and $\gamma\in[1,\infty)$
and has the property
$$\pi_\rho=\frac{p_\rho}{\rho}.$$

If we assume \defm{irrotationality}
$$v^i_j=v^j_i$$
(where $i,j=1,\dotsc,d$), then the Euler equations are reduced to potential flow:
\begin{alignat*}{1}
    \vec v &= \nabla_{\vec x}\phi
\end{alignat*}
for some scalar \defm{potential}\footnote{We consider simply connected domains; otherwise $\phi$ might be multivalued.} 
function $\phi$. For smooth flows, substituting this into \myeqref{eq:v} yields, for $i=1,\dotsc,d$,
\begin{alignat}{1}
    0 
    &= \phi_{it} + \nabla\phi_i\cdot\nabla\phi + \pi(\rho)_i = \big(\phi_t + \frac{|\nabla\phi|^2}{2} + \pi(\rho)\big)_i. \notag
\end{alignat}
Thus, for some constant $A$, 
\begin{alignat}{1}
    \rho &= \pi^{-1}(A-\phi_t-\frac{|\nabla\phi|^2}{2}). \myeqlabel{eq:rhoA}
\end{alignat}
Substituting this into \myeqref{eq:rhodiv} yields a single second-order quasilinear hyperbolic equation, the \defm{potential flow} equation, for a scalar field $\phi$:
\begin{alignat}{1}
	\big(\rho(\phi_t,|\nabla\phi|)\big)_t+\nabla\cdot\big(\rho(\phi_t,|\nabla\phi|)\nabla\phi\big) &= 0. \myeqlabel{eq:potflow-divform}
\end{alignat}
Henceforth we omit the arguments of $\rho$. Moreover we eliminate $A$ with the substitution 
\begin{alignat}{1}
	&A\leftarrow 0,\qquad\phi(t,\vec x)\leftarrow\phi(t,\vec x)-tA\notag
\end{alignat}
(so that $\phi_t\leftarrow\phi_t-A$). Hence we use
\begin{alignat}{1}
	\rho &= \pi^{-1}(-\phi_t-\frac{1}{2}|\nabla\phi|^2)\myeqlabel{eq:rho}
\end{alignat}
from now on.

Using $c^2=p_\rho$ and 
\begin{alignat}{1}
  (\pi^{-1})' &= (\pi_\rho)^{-1} = (\frac{p_\rho}{\rho})^{-1}=\frac{\rho}{c^2} \myeqlabel{eq:Dpiinv}
\end{alignat}
the equation can also be written in nondivergence form:
\begin{alignat}{1}
    \phi_{tt} + 2\nabla\phi_t\cdot\nabla\phi + \sum_{i,j=1}^d\phi_i\phi_j\phi_{ij} - c^2\Delta\phi &= 0 \myeqlabel{eq:potential-flow}
\end{alignat}
\myeqref{eq:potential-flow} is hyperbolic (as long as $c>0$).
For polytropic pressure law the local sound speed $c$ is given by
\begin{alignat}{1}
	c^2 &= c_0^2 + (\gamma-1)(-\phi_t-\frac{1}{2}|\nabla\phi|^2). \myeqlabel{eq:cs-uspf}
\end{alignat}

Our initial data is self-similar: it is constant along rays emanating from $\vec x=(0,0)$.
Our domain $V$ is self-similar too: it is a union of rays emanating from $(t,x,y)=(0,0,0)$.
In any such situation it is expected --- and confirmed by numerical results --- that the solution is self-similar as well, i.e.\ that
$\rho,\vec v$ are constant along rays $\vec x=t\vec\xi$ emanating from the origin. 
Self-similarity corresponds to the ansatz
\begin{alignat}{1}
	\phi(t,\vec x) &:= t\psi(\vec\xi),\qquad\vec\xi := t^{-1}\vec x \myeqlabel{eq:psi-phi}.
\end{alignat}
Clearly, $\phi\in C^{0,1}(\Omega)$ if and only if $\psi\in C^{0,1}(\complement W)$.
This choice yields
\begin{alignat}{1}
	\vec v(t,\vec x) &= \nabla\phi(t,\vec x) = \nabla\psi(t^{-1}\vec x), \notag\\
	\rho(t,\vec x) &= \pi^{-1}(-\phi_t-\frac{1}{2}|\nabla\phi|^2) 
	= \pi^{-1}(-\psi+\vec\xi\cdot\nabla\psi-\frac{1}{2}|\nabla\psi|^2).\notag
\end{alignat}
The expression for $\rho$ can be made more pleasant (and independent of $\vec\xi$) by using
\begin{alignat}{1}
	\chi(\vec\xi) &:= \psi(\vec\xi)-\frac{1}{2}|\vec\xi|^2;\notag
\end{alignat}
this yields
\begin{alignat}{1}
	\rho &= \pi^{-1}(-\chi-\frac{1}{2}|\nabla\chi|^2). \myeqlabel{eq:rhoeq}
\end{alignat}
$\nabla\chi=\nabla\psi-\vec\xi$ is called \defm{pseudo-velocity}. 

\myeqref{eq:potflow-divform} then reduces to
\begin{alignat}{1}
	\nabla\cdot(\rho\nabla\chi)+2\rho &= 0 \myeqlabel{eq:chi-divform}
\end{alignat}
(or $+d\rho$, in $d$ dimensions)
which holds in a distributional sense.
For smooth solutions we obtain the non-divergence form
\begin{alignat}{1}
	(c^2I-\nabla\chi\nabla\chi^T):\nabla^2\chi
	= (c^2-\chi_\xi^2)\chi_{\xi\xi}-2\chi_\xi\chi_\eta\chi_{\xi\eta}+(c^2-\chi_\eta^2)\chi_{\eta\eta}
	&= |\nabla\chi|^2-2c^2 \myeqlabel{eq:chi} 
\end{alignat}
Another convenient form is
\begin{alignat}{1}
	(c^2I-\nabla\chi\nabla\chi^T):\nabla^2\psi 
	&= (c^2-\chi_\xi^2)\psi_{\xi\xi}-2\chi_\xi\chi_\eta\psi_{\xi\eta}+(c^2-\chi_\eta^2)\psi_{\eta\eta} 
	= 0. \myeqlabel{eq:psi}
\end{alignat}
Here, \myeqref{eq:cs-uspf} for polytropic pressure law yields 
\begin{alignat}{1}
	c^2 &= c_0^2 + (\gamma-1)(-\chi-\frac{1}{2}|\nabla\chi|^2) \myeqlabel{eq:css}
\end{alignat}

\begin{remark}
	\mylabel{rem:symmetries}%
	\myeqref{eq:chi-divform} inherits a number of symmetries from \myeqref{eq:rhodiv}, \myeqref{eq:mom}: 
	\begin{enumerate}
	\item It is invariant under rotation.
	\item It is invariant under reflection.
	\item It is invariant under translation in $\vec\xi$, which is not as trivial as translation in $\vec x$:
	it corresponds to the Galilean transformation $\vec v\leftarrow \vec v+\vec v_0$, 
	$\vec x\leftarrow \vec x-\vec v_0t$ (with constant $\vec v_0\in\R^d$) in $(t,\vec x)$ coordinates.
	This is sometimes called \defm{change of inertial frame}.
	\end{enumerate}
\end{remark}

\myeqref{eq:chi} is a PDE of mixed type. The type
is determined by the \defm{(local) pseudo-Mach number}
\begin{alignat}{1}
    L &:= \frac{|\nabla\chi|}{c}, \myeqlabel{eq:L}
\end{alignat}
with $0\leq L<1$ for elliptic (pseudo-subsonic), $L=1$ for parabolic (pseudo-sonic), $L>1$ for hyperbolic (pseudo-supersonic) regions.

While velocity $\vec v$ is motion
relative to space coordinates $\vec x$, pseudo-velocity 
$$\vec z:=\nabla\chi$$
is motion relative to similarity coordinates
$\vec\xi$ \emph{at time $t=1$}.

The simplest class of solutions of \myeqref{eq:chi} are the \defm{constant-state solutions}: $\psi$ affine in $\vec\xi$, 
hence $\vec v$, $\rho$ and $c$ constant. They are elliptic in a circle centered in $\vec\xi=\vec v$ with radius $c$, parabolic
on the boundary of that circle and hyperbolic outside.

\begin{convention}
	If we study a function called (e.g.) $\tilde\chi$, then $\tilde\psi$, $\tilde\rho$, $\tilde L$ etc.
	will refer to the quantities computed from it as $\psi$, $\rho$, $L$ are computed from $\chi$
	(e.g.\ $\tilde\psi=\tilde\chi+\frac{1}{2}|\vec\xi|^2$). We will
	tacitly use this notation from now on.
\end{convention}

\subsection{Potential flow shocks}

\mylabel{section:potf-shocks}

Consider a ball $U$ and a simple smooth curve $S$ so that $U=U^u\cup S\cup U^d$ where
$U^u,U^d$ are open, connected, and $S,U^u,U^d$ disjoint.
Consider $\chi:U\rightarrow\R$ so that $\chi=\chi^{u,d}$ in $U^{u,d}$ where 
$\chi^{u,d}\in \spC^2(\overline{U^{u,d}})$. 

$\chi$ is a weak solution of \myeqref{eq:chi-divform} if and only if it is a strong solution in each point of $U_-$ and $U_+$ and 
if it satisfies the following conditions in each point of $S$: 
\begin{alignat}{1}
	\chi^u &= \chi^d, \myeqlabel{eq:chijump} \\
	\vec n\cdot(\rho^u\nabla\chi^u-\rho^d\nabla\chi^d) &= 0 \myeqlabel{eq:momjump}
\end{alignat}
Here $\vec n$ is a normal to $S$.

\myeqref{eq:chijump} and \myeqref{eq:momjump} are the \defm{Rankine-Hugoniot} conditions for self-similar potential flow
shocks. They do not depend on $\vec\xi$ or on the shock speed explicitly; these quantities
are hidden by the use of $\chi$ rather than $\psi$. The Rankine-Hugoniot conditions are derived in the same way as
those for the full Euler equations (see \cite[Section 3.4.1]{evans}).

Note that \myeqref{eq:chijump} is equivalent to 
\begin{alignat}{1}
	\psi^u &= \psi^d. \myeqlabel{eq:psijump}
\end{alignat}
Taking the tangential derivative of \myeqref{eq:chijump} resp.\ \myeqref{eq:psijump} yields
\begin{alignat}{1}
	\frac{\partial\chi^u}{\partial t} &= \frac{\partial\chi^d}{\partial t} \myeqlabel{eq:chitan}, \\
	\frac{\partial\psi^u}{\partial t} &= \frac{\partial\psi^d}{\partial t} \myeqlabel{eq:psitan}.
\end{alignat}
The shock relations imply that the tangential velocity is continuous across shocks.

Define $(z^x_u,z^y_u):=\vec z_u:=\nabla\chi^u$ and $(v^x_u,v^y_u):=\vec v_u:=\nabla\psi^u$.
Abbreviate $z^t_u:=\vec z_u\cdot\vec t$, $z^n_u:=\vec z_u\cdot\vec n$, and same for $v$ instead of $z$.
Same definitions for $d$ instead of $u$.
We can restate the shock relations as
\begin{alignat}{1}
	\rho_uz^n_u &= \rho_dz^n_d, \myeqlabel{eq:steady-continuity} \\
	z_u^t &= z_d^t. \myeqlabel{eq:chitan-z}
\end{alignat}
Using the last relation, we often write $z^t$ without distinction.

The \defm{shock speed} is $\sigma=\vec\xi\cdot\vec n$, where $\vec\xi$ is any point on the shock. 
A shock is \defm{steady} in a point if its tangent passes through the origin.
We can restate \myeqref{eq:steady-continuity} as
\begin{alignat}{1}
	\rho_uv^n_u-\rho_dv^n_d &= \sigma(\rho_u-\rho_d)\notag
\end{alignat}
which is a more familiar form.

We focus on $\rho_u,\rho_d>0$ from now on, which will be the case in all circumstances. 
If $\rho_u=\rho_d$ in a point, we say the shock \defm{vanishes}; in this case $z^n_d=z^n_u$ in that point, by \myeqref{eq:chitan-z}. 
In all other cases $z^n_d,z^n_u$ must have equal sign by \myeqref{eq:chitan-z}; we fix $\vec n$ so that $z^n_d,z^n_u>0$. 
This means the normal points \defm{downstream}. 
The shock is \defm{admissible} if and only if $\rho_u\leq\rho_d$ which is equivalent to $z^n_u\geq z^n_d$.

A shock is called \defm{pseudo-normal} in a point $\vec\xi$ if $z^t=0$ there.
For $\vec\xi=0$, this means that the shock is \defm{normal} ($v^t=0$), but for $\vec\xi\neq 0$ normal and pseudo-normal are not always equivalent.

It is good to keep in mind that for a \emph{straight} shock, $\rho_d$ and $\vec v_d$ are constant if $\rho_u$ and $\vec v_u$ are.
Obviously $\vec z_d$ may vary in this case.

We will need two detailed results.

\begin{proposition}
	\mylabel{prop:shockpolar}%
	Consider a fixed point on a shock with upstream density $\rho_u$ and pseudo-velocity $\vec z_u$ held fixed
	while we vary the normal. Define $\nva:=\measuredangle(\vec z_u,\vec n)$.
	$\rho_d$ is strictly decreasing in $|\nva|$, whereas
	$L_d,|\vec z_d|$ are strictly increasing. $c_d$ is strictly decreasing for $\gamma>1$, constant otherwise.
	Moreover
	\begin{alignat}{1}
		(\partial_\nva\vec v_d)\cdot\vec n = (\partial_\nva\vec z_d)\cdot\vec n &= z^t\Big(\frac{\partial z^n_d}{\partial z^n_u}-1\Big), \myeqlabel{eq:DvdxDnva} \\
		(\partial_\nva\vec v_d)\cdot\vec t = (\partial_\nva\vec z_d)\cdot\vec t &= z^n_d-z^n_u. \myeqlabel{eq:DvdyDnva} 
	\end{alignat}

	If $\vec z_u=(z^x_u,0)$ with $z^x_u>0$, then $z^x_d$ is increasing in $|\nva|$.
\end{proposition}
\begin{proof}
	This is \pmc{Proposition \pmref{prop:shockpolar}}.
\end{proof}

\begin{proposition}
	\mylabel{prop:vdzero}%
	Consider a straight shock with $v^x_u=0$, $v^y_u<0$ and downstream normal $\vec n=(\sin\nva,-\cos\nva)$
	through $\vec\xi=(0,\eta)$.
	For every $\nva\in(-\frac{\pi}{2},\frac{\pi}{2})$ there is a unique $\eta=\eta^*_0\in\R$ so that $v^y_d=0$. $\eta^*_0$ and the corresponding
	downstream data are analytic functions of $\nva$. $\eta^*_0$ is strictly increasing in $|\nva|$.

	For the shock passing through $(0,\eta^*_0)$, 
	let $\vxit_L$ and $\vxit_R$ be the two points with $L_d=\sqrt{1-\epsilon}$.
	These points are analytic functions of $\nva$.
	$L^n_u$, $\rho_d$ and $z^n_u$ are increasing functions\footnote{All of these are independent of the location along the (straight) shock.} 
	of $\nva$; 
	$v^x_d$ and $L^n_d$ are decreasing functions of $\nva$.
	For $\nva\in[0,\frac{\pi}{2})$, $\etat_L$ is a strictly decreasing function of $\nva$ with range $(\underline\eta^*_L,\overline\eta^*_0]$,
	where $\overline\eta^*_0$ is the $\eta^*_0$ for $\nva=0$, and $\underline\eta^*_L$ is some \emph{negative} constant.
\end{proposition}
\begin{proof}
	This is \pmc{Proposition \pmref{prop:vdzero}}.
\end{proof}

\subsection{Envelope}

\mylabel{section:envelope}

Many techniques in this paper are similar to the construction in \cite{elling-liu-pmeyer-arxiv}; Section \pmref{section:approach} in loc.cit.\ is a good overview.
However, in \pmc{Proposition \pmref{prop:rho}}, 
a lower bound for the shock strength is obtained by a delicate argument using the density. 
Although this argument would reproduce the results of \cite{chen-feldman-selfsim-journal} (namely RR existence for $\theta\approx\frac\pi2$), it cannot
prove the main new contribution of this paper: existence (at least in some cases like $\gamma=5/3$, $M_I=1$) of RR
for $\theta\approx\theta_s$ (with $\theta>\theta_s$), where $\theta_s$ is the smallest $\theta$ allowed by the sonic criterion 
(see Section \myref{section:refl-selfsim}).

For this goal, a new idea is needed: as we will show, the curved portion $S$ of the reflected shock in Figure \myref{fig:rrsmr} left
has an elliptic region of potential 
flow on its right (downstream) side, hence\ downstream pseudo-Mach number $L_d\leq 1$ everywhere. 
Such a shock cannot vanish until it reaches the circle of radius $c_I$ around $\vec v_I$; moreover $L_d\leq 1$ is a constraint on the possible shock tangents,
so that the shock cannot reach the circle quickly. It is bounded away from the circle by the \defm{envelope}:

\begin{figure}
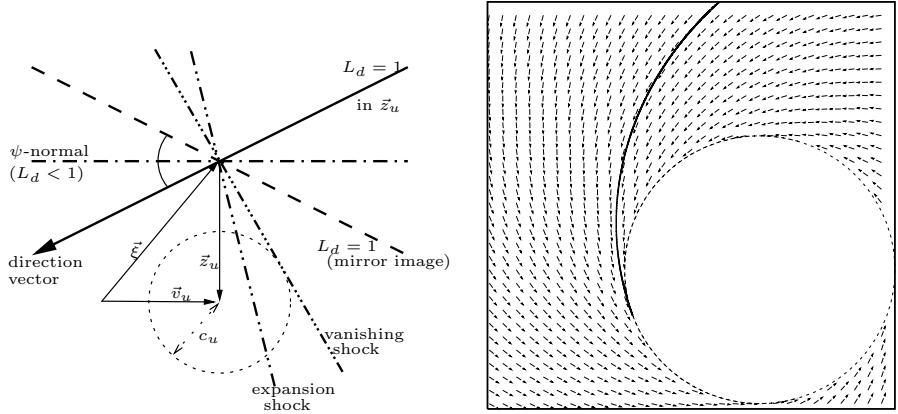

\input{envshock.pstex_t}\hbox to10mm{}\includegraphics[width=.45\linewidth]{envfield.epsi}
\caption{Left: through each $\vec\xi$ farther than $c_u$ from $\vec v_u$ there are exactly two straight shocks (solid, dashed) with $L_d=1$,
	mirror images of each other. The shocks with $L_d\leq 1$ are between them (indicated by arc left of $\vec\xi$).
	The solid lines define the direction field whose integral curves are ``envelopes''.
	Right: no shock with $L_d\leq 1$ can approach $\vec v_u$ faster (in counterclockwise direction) than the counterclockwise envelope.}
\mylabel{fig:env}
\end{figure}

\begin{definition}
	\mylabel{def:envelope}%
	Given constant upstream velocity $\vec v_u$ and sound speed $c_u$.
	Consider a shock through a point $\vec\xi$ with $|\vec z_u|=|\vec v_u-\vec\xi|>c_u$. 
	As shown in Proposition \myref{prop:shockpolar}, 	
	$L_d$ is strictly increasing in $|\beta|$ where $\beta=\measuredangle(\vec z_u,\vec n)\in(-\pi,\pi]$ is the counterclockwise angle from $\vec z_u$ to $\vec n$.

	There are exactly two shock normals so that $L_d=1$. They are mirror-images of each other under reflection across the line with tangent
	$\vec z_u$ through $\vec\xi$ (see Figure \myref{fig:env} left). Consider the one with $\beta>0$; its tangent spans the 
	solid line on Figure \myref{fig:env} left. The tangents for different $\vec\xi$ form a direction field.
	The \defm{counterclockwise envelope} is defined to be a maximal integral curve of that direction field (see Figure \myref{fig:env} right).
\end{definition}

We can parametrize the envelope (like other smooth shocks) in polar coordinates $(r,\phi)$ centered in $\vec v_u$,
by a function $\phi\mapsto r^*(\phi)$ (because the shock relations do not admit shocks with a tangent passing through $\vec v_u$).
The counterclockwise envelope satisfies an ODE of the form 
\begin{alignat}{1}
	\frac{\partial r^*}{\partial\phi}(\phi)=-f(r^*(\phi))  \mylabel{eq:envelope-ode}
\end{alignat}
for some analytic $f$. 

We will not need the fact, but explicit formulas for $f$ can be derived. For example for $\gamma>1$,
\begin{alignat}{1}
	f(r) &= r\sqrt{\frac{
		1-\frac{\gamma+1}{\gamma-1+2(r/c_u)^{-2}}\cdot\Big(\frac{\gamma+1}{2+(\gamma-1)(r/c_u)^2}\Big)^\frac{2}{\gamma-1} 
	}{
		\frac{\gamma+1}{\gamma-1+2(r/c_u)^{-2}}-1
	}}
	\mylabel{eq:envelope-ode-explicit}
\end{alignat}
Moreover it can be shown that the envelope always reaches the circle, 
meeting it in a point where the envelope is $C^1$, but not more regular, and tangent to the circle; it cannot be continued beyond that point. 

\begin{proposition}
	\mylabel{prop:shock-envelope}%
	Let some smooth shock be parametrized as $\phi\mapsto r(\phi)$; let the envelope be parametrized by $\phi\mapsto r^*(\phi)$.
	Assume that $L_d<1$ in every point of the shock. If $r(\phi_0)\geq r^*(\phi_0)$ for some $\phi_0$, then 
	$r(\phi)>r^*(\phi)$ for $\phi>\phi_0$.
	If instead $L_d>1$ in every point of the shock, then
	$r(\phi)<r^*(\phi)$ for $\phi>\phi_0$.
\end{proposition}
\begin{proof}
	Our discussion above can be restated as follows: $L_d<1$ for the shock means $-\beta^*<\beta<\beta^*$ where $\beta^*$ is
	the $\beta$ for the envelope. Hence
	$$|\frac{\partial r}{\partial\phi}| < f(r(\phi)).$$
	In particular
	$$\frac{\partial r}{\partial\phi} > -f(r(\phi)).$$
	Since $f$ is smooth, in particular Lipschitz, the invariant region theorem shows that the shock cannot meet the envelope for $\phi>\phi_0$.
\end{proof}

In Proposition \myref{prop:shockenv} we will exploit this fact to bound the curved portion of the reflected shock away from the downstream
wall and to ensure its uniform strength.

\subsection{Sonic criterion}

We focus on the classical case of vertical incident shocks.
In some cases, Theorem \myref{th:elling-rrefl} allows us to construct a regular reflection pattern
like Figure \myref{fig:rrsmr} left for every myref{section:refl}). 
As $\theta\downarrow\theta_s$, the dotted$\theta>\theta_s$ near $\theta_s$, where $\theta_s$ is the smallest $\theta$ allowed
by the sonic criterion (see Section \ parabolic arc in Figure \myref{fig:rrsmr} left approaches the reflection point.

To check whether the envelope condition is satisfied for a particular choice of $\theta$ and incident shock, it suffices to find
the reflected shock and $\vec\xi_C$ on it (see Theorem \myref{th:elling-rrefl}) and to integrate the ODE \myeqref{eq:envelope-ode} 
defining the envelope. Although the ODE is trivially separable, the resulting integral and nonlinear algebraic equation
do not have an explicit solution except for special values of $\gamma$ (see \myeqref{eq:envelope-ode-explicit}). 
Numerical integration is needed to check whether the envelope
meets $\hat B$ or the circle with center $\vec v_I$ and radius $c_I$ before it meets $\hat A$.

\begin{figure}
\input{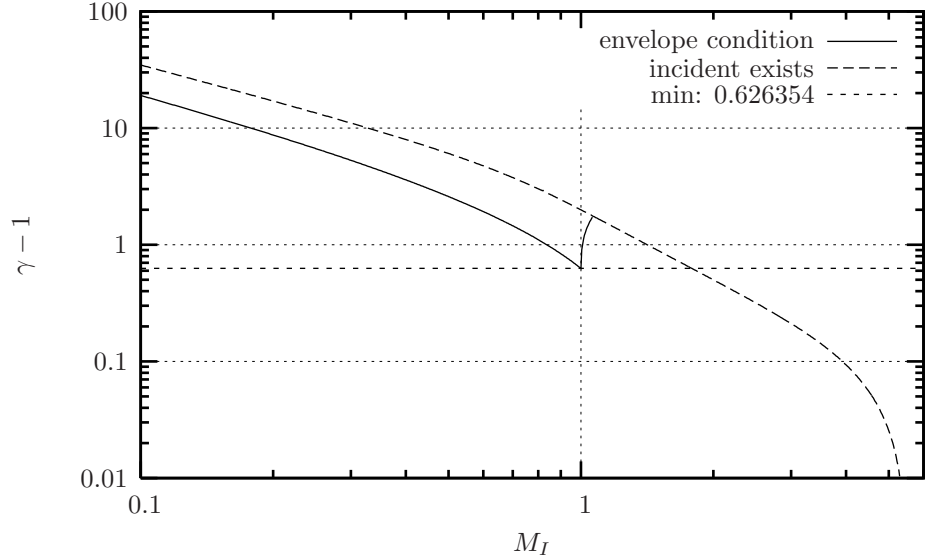}
\caption{For $\beta_Q=0$ (vertical incident shocks) and the set of $M_I,\gamma$ enclosed below the dashed and above the solid line, 
solutions can be constructed for all $\theta\in(\theta_s,\frac\pi2]$.}
\mylabel{fig:good}
\end{figure}

In Figure \myref{fig:good}, we consider arbitrary $\gamma\in[1,\infty)$ and $M_I\in(0,\infty)$ while fixing $\theta=\theta_s$.
Values of $\gamma$ and $M_I$ above the dashed curve do not admit a vertical incident shock with zero velocity
in the $Q$ region (a similar phenomenon occurs in the full Euler equations). 
Values below both solid and dashed curve violate
the envelope condtion. 
Values between solid and dashed curve do have an incident shock
as well as a reflected shock that satisfies the envelope condition. 

The smallest possible $\gamma$ in that feasible region is $\gamma=1.626354...$ with $M_I=1$. In particular
the monatomic gas case $\gamma=5/3$ is covered, whereas $\gamma=7/5$ or $\gamma=4/3$ are not covered. 
(However, the latter values are also possible if we allow non-vertical incident shocks.)
For $\gamma=5/3$, $M_I=1$ we have $\theta_s=55.4583...^\circ$; for $\theta=\theta_s$ the envelope meets $\hat A$ in the point $(-0.000012...,0)$, 
just enough to avoid $\hat B$ and the circle.

While the proof of Theorem \myref{th:elling-rrefl} itself is rigorous, checking the envelope condition is done numerically here,
i.e.\ not a mathematical
proof in the strict sense. However, the shock relations form a small system of nonlinear algebraic equations and the envelope
is defined by \myeqref{eq:envelope-ode}, a scalar nonlinear ODE which is benign except for a mild singularity as $r\downarrow 1$.
The numerical methods for these types of equations are well-understood and a complete convergence theory and error analysis is available ---
which is not at all the case for the full Euler or potential flow PDE. 
Another option is to study rigorous proofs in various asymptotic limits such as $M_I\downarrow 0$, $\gamma\uparrow\infty$.
Moreover the envelope condition is most likely unnecessary since regular reflection up to $\theta=\theta_s$ is observed in numerics for many other
values of $\gamma$ and $M_I$ as well. Since we expect that the condition will be eliminated by further research, it makes little
sense to strive for absolute rigour at this point.

\section{Construction of the flow}
\mylabel{section:ellreg}

\newcommand{\fusp}{\mathcal{F}}
\newcommand{\cfusp}{\overline\fusp}
\newcommand{\IT}{\mathcal{K}}
\newcommand{\BL}{\mathcal{L}}

The elliptic region is constructed as follows: we define a function set $\fusp$ by imposing many constraints on a weighted H\"older space
$\spC^{2,\alpha}_\beta$ (weighted to account for loss of regularity in the corners). An iteration $\IT:\fusp\rightarrow\spC^{2,\alpha}_\beta$
is constructed so that its fixed points solve the PDE and boundary conditions for the elliptic region (see Remark \myref{rem:fp}). $\fusp$ and $\IT$ depend on
several parameters like $\gamma$, collected in a parameter vector $\lambda$. To show that $\IT$ has a fixed point for all $\lambda$,
we use Leray-Schauder degree theory. 

Most of the effort is spent on showing that $\IT$ does not have fixed points on $\partial\fusp$, which implies that $\IT$ has the
same Leray-Schauder degree for all $\lambda$. As $\partial\fusp$ is defined by constraints
in the form of inequalities with continuous sides, 
this is achieved by showing that a fixed point satisfies the \emph{strict} version of each inequality ($<$ instead
of $\leq$). 

A major technical difficulty are the parabolic arcs (dotted arc in Figure \ref{fig:rrsmr} left) where self-similar potential flow \myeqref{eq:psi}
degenerates from elliptic to parabolic. This problem has been solved in \cite{elling-liu-pmeyer-arxiv} (and, by different techniques, in
\cite{chen-feldman-selfsim-journal}), by modifying the arc to be slightly elliptic, with boundary condition $L^2=1-\epsilon$, and obtaining
estimates uniform in $\epsilon$.

For a particular choice of $\lambda$ the problem is much simpler (see Figure \ref{fig:unperturbed}). In that case an explicit solution can be
given and shown to be unique and have nonzero Leray-Schauder index. This implies that $\IT$ has nonzero degree, hence at least one fixed point,
for \emph{every} $\lambda$. The fixed point is extended to a solution on the entire domain by adding the hyperbolic regions and
interface shocks. Using the $\epsilon$-uniform estimates as well as compactness, we can pass to the limit $\epsilon\downarrow 0$ to obtain
a solution of our problem.

\subsection{Parameter set and definitions}

\mylabel{section:parmset}

\begin{figure}
\input{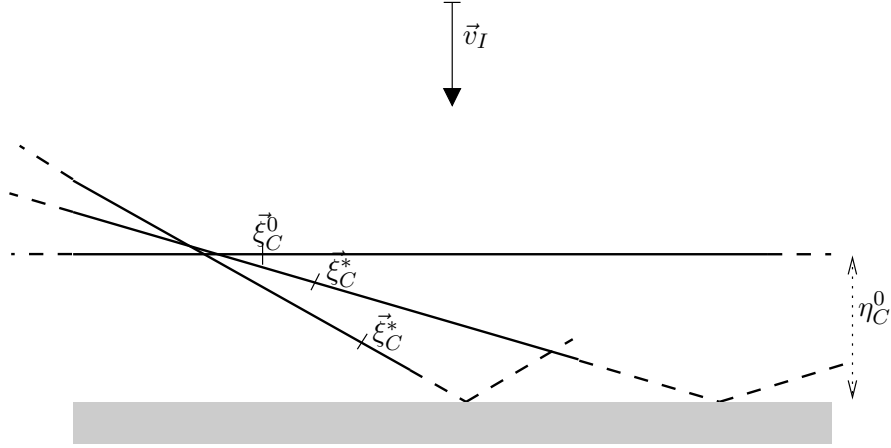}
\caption{Perturbation from the trivial case of $R$ parallel to the wall.}
\mylabel{fig:framerefl}
\end{figure}

\begin{figure}
\input{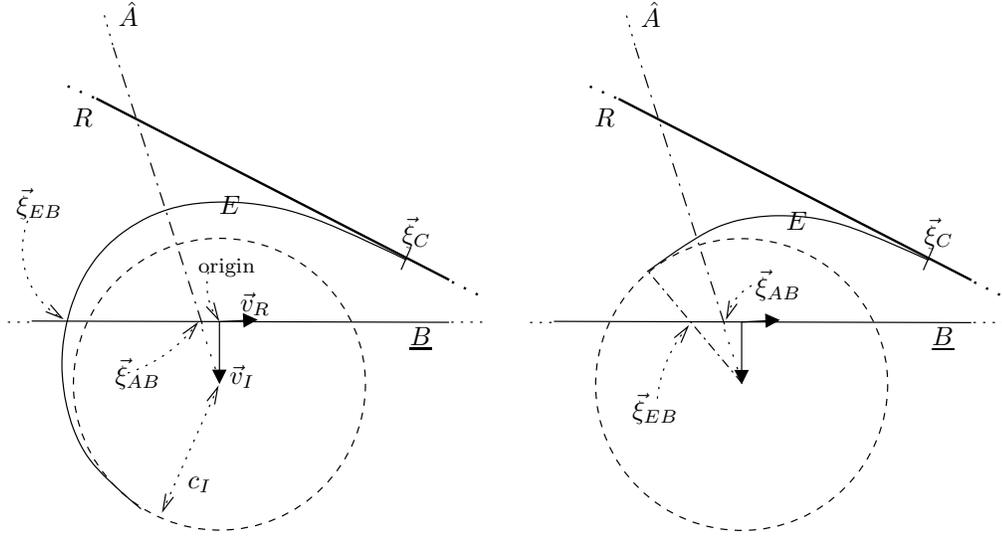}
\caption{$\hat A$ is chosen so that (1) $E$ reaches it before $\fullB$ or the dashed circle, and (2) it forms an angle $\leq90^\circ$ with $R$.}
\mylabel{fig:fullB}
\end{figure}

\begin{figure}
\input{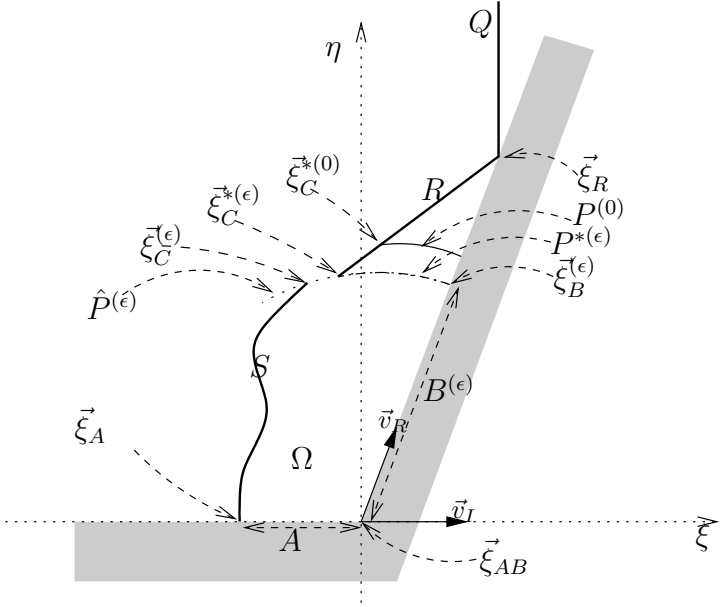}
\caption{To avoid degeneracy, we impose a ``slightly elliptic'' boundary condition, $L^2=1-\epsilon$ for
$\epsilon>0$, on $P^{(\epsilon)}$. 
The shock $S$ is free, along with the endpoints $\vec\xi_A$ and $\vec\xi_C^{(\epsilon)}$ which may slide freely on $\hat A$
resp.\ $\hat P^{(\epsilon)}$. But for fixed points
$\vec\xi_C^{(\epsilon)}$ can be shown to be close to $\vec\xi^{*(\epsilon)}$, hence to $\vec\xi^{*(0)}$.}
\mylabel{fig:frame0}
\end{figure}

Instead of working in the setting of Theorem \myref{th:elling-rrefl}, it will be convenient to choose parameters in a different way.

Choose $\rho_I,c_I>0$. Note that we may fix $\rho_0$ and $c_0$ in the pressure law \myeqref{eq:p-polytropic} separately; however,
given these constants (and $\gamma$), every other $c$ is a function of $\rho$ only (and vice versa).

Let $\epsilon\geq0$ be sufficiently small for the following. Consider a vertical downward velocity $\vec v_I$ onto a solid wall $\fullB$ (see Figure
\myref{fig:framerefl}). 
According to Proposition \myref{prop:vdzero}, there is exactly one straight shock
with upstream velocity $\vec v_u=\vec v_I$ and sound speed $c_u=c_I$ so that $\vec v_d=0$; that shock is horizontal. 
Let $\eta_C^0>0$ be its vertical coordinate. Of the two points on that shock with $L_d=\sqrt{1-\epsilon}$, let $\vec\xi^0_C=(\xi^0_C,\eta^0_C)$ be the 
right one. 
By the same proposition, the shock belongs to a smooth one-parameter family of shocks, each called \defm{R shock}, 
parametrized by $\eta_C^*\in(0,\eta_C^0]$,
so that $v^y_d=0$ and so that $\vec\xi_C^*=(\xi_C^*,\eta_C^*)$ is the right $L_d=\sqrt{1-\epsilon}$ point. 
Define $M^y_I\in[-1,0)$ to be $v_I^y/c_I$ in these coordinates. 
Note that \myeqref{eq:MyI-one} rules out $M^y_I<-1$.
Let $\vec v_R=\vec v_d$ be the downstream velocity of the $R$ shock.

It is not clear whether there is an incident shock $Q$ matching each reflected shock $R$. In fact for $\eta^*_C=\eta^0_C$, the $R$ shock does
not even meet $\fullB$, so clearly there is no RR.
However, for the construction of the elliptic region, a $Q$ shock or reflection pattern are not needed.

To complete the situation of Theorem \myref{th:elling-rrefl}, a wall $\hat A$ is needed. To satisfy the slip boundary condition 
$(\vec v_I-\vec\xi)\cdot\vec n=0$ on $\hat A$, necessarily the extension of $\hat A$ to a line has to pass through $\vec v_I$. 
We fix $\hat A$ by choosing $\vec\xi_{AB}$ on $\fullB$.

Let $E$ be the counterclockwise envelope starting in $\vec\xi_C$. 
If $E$ meets $\fullB$ before it meets the circle with center $\vec v_I$ and radius $c_I$ (Figure \myref{fig:fullB} left), 
let $\vec\xi_{EB}$ be that point. Otherwise (Figure \myref{fig:fullB} right) 
take the line through $\vec v_I$ and the meeting point of $E$ and circle, and let
$\vec\xi_{EB}$ be its intersection with $\fullB$. We allow 
\begin{alignat}{1}
	\xi_{AB}&\in(\xi_{EB},v^x_R] \mylabel{eq:xiAB}
\end{alignat}
(and $\eta_{AB}=0$ obviously). This constraint ensures that (1) the envelope meets $\hat A$ first, while (2) $R$ and $A$ form a sharp
or right angle.

Given $\vec\xi_{AB}$ we let $\hat B$ be the part of $\fullB$ \emph{right} of $\vec\xi_{AB}$. $\hat A$ is the half-line upwards starting in $\vec\xi_{AB}$
whose extension passes through $\vec v_I$. Let $\vec n_A$ be the unit normal of $A$ pointing left, $\vec n_B$ the unit normal of $\hat B$ pointing down.
Let $\vec n_R$ be the downstream (hence downwards) unit normal of the $R$ shock. For each $\vec n_?$, $\vec t_?$ is always the corresponding
unit tangent in \emph{counterclockwise} direction.

\begin{remark}
	\mylabel{rem:Lambda-full}%
	Every local RR pattern that satisfies the conditions of Theorem \myref{th:elling-rrefl} is covered by the parameter ranges defined above.
\end{remark}

$\rho_I$ and $\vec v_I$ define a potential $\psi^I$ for the $I$ region:
$$\psi^I(\vec\xi)=-\pi(\rho_I)-\frac{|\vec v_I|^2}{2}+\vec v_I\cdot\vec\xi.$$
Similar potentials $\psi^R$ and $\psi^Q$ (if an incident shock $Q$ exists) 
are defined by $\rho_R,\vec v_R$ and $\rho_Q,\vec v_Q$.

Now we use Remark \myref{rem:symmetries}: invariance under translation. Translation in self-similar coordinates corresponds to a change of inertial
frame, i.e.\ to adding a constant velocity to all $\vec v,\vec\xi$. Moreover we may rotate by Galilean invariance. This changes Figure \myref{fig:fullB}
to Figure \myref{fig:frame0} which has the coordinates in which we originally posed the self-similar reflection problem.

Let $P^{*(\epsilon)}$ be the circle arc centered in $\vec v_R$ with radius $c_R\cdot\sqrt{1-\epsilon}$ (see Figure \myref{fig:frame0}, where
the coordinates have been changed), passing from $\vec\xi_B^{(\epsilon)}$ on $\hat B$
counterclockwise to $\vec\xi_C^{*(\epsilon)}$ on $R$, excluding the endpoints. (We omit the superscript $\epsilon$ if it is clear from the context.)
$\vxit_C$ will be called the \defm{expected} corner location. 
Let $B^{(\epsilon)}$ be the part of $\hat B$ from $\vec\xi_{AB}$ to $\vec\xi_B$ (excluding the endpoints).

Take $\vec n_R$, $\vec n_Q$ to be the downstream unit normals of the shocks $R,Q$ ($\vec n_R$ points towards $\hat B$). 
Let $\vec n_A,\vec n_B$ be outer unit normals of $\hat A,\hat B$, i.e.\ pointing away from the gas-filled sector $V$ enclosed by $\hat B,\hat A$.

We choose an extended arc $\hat P$ that overshoots $\vxit_C$ by an angle $\delta_{\hat P}>0$, which we choose continuous in $\gamma,\xi_{AB},\eta^*_C$.
The particular $\delta_{\hat P}$ is not important, but it may not depend on $\epsilon$, and $\hat P$ may not have a horizontal tangent
in Figure \myref{fig:onion} coordinates.

$P^*$, $\hat P$, and later $P$, are called \defm{quasi-parabolic arc} (or \defm{parabolic arcs}, by abuse of terminology, 
or short \defm{arcs}).

\paragraph{Parameter set}

The Definitions \myref{def:Lambda}, \myref{def:b} and \myref{def:fusp} 
use many constants and other objects that will be fixed later on. 
In all of these cases, an upper (or lower) bound for each constant is found.
Whenever we say ``for sufficiently small constants'' (etc.), we mean that bounds for them are adjusted.
To avoid circularity, it is necessary to specify which bounds may depend on the values of which other bounds.
In the following list, bounds on a constant may only depend on bounds on \emph{preceding} constants.
\begin{alignat}{1}
	& \delta_{\hat P},
	C_L,C_\eta,\delta_{SB},\delta_{Cc},\delta_{P\sigma},\delta_{Pn},\delta_d,\delta_\rho,\delta_{Lb},\notag\\
	&\qquad C_{Pt},C_{vtR},C_{vnA},C_{Sn},\delta_{vtA},\delta_{vnB},\delta_o,C_d,\epsilon,C_{\spC},r_I,\alpha,\beta. \myeqlabel{eq:constlist}
\end{alignat}
The constants $C_{\spC},r_I,\alpha,\beta$ may depend on $\epsilon$ itself, not just on an upper bound. 
$r_I$ may also depend on $\po$.
The reader may convince himself that the remainder of the paper does respect this order.

The parameters $\gamma$, $\eta^*_C$ and $\xi_{AB}$ used in Leray-Schauder
degree arguments will be restricted to compact
sets below so that any constant that can be chosen continuous in them might as well be taken independent of them.
Dependence on other parameters like $\rho_I$ will not be pointed out explicitly. 

Constants $\delta_?$ as well as $\alpha,\beta,r_I,\epsilon$ are meant to be small and positive, constants $C_?$ are meant to be large 
and finite.

\begin{definition}
	\mylabel{def:Lambda}%
	For the purposes of degree theory we define a restricted parameter set
	$$\Lambda:=\Big\{\lambda=(\gamma,\eta^*_C,\xi_{AB}):
		\gamma\in[1,\overline\gamma],~
		\eta^*_C\in[\underline\eta^*_C,\overline\eta^*_C],
		\xi_{AB}\in[\underline\xi_{AB},\overline\xi_{AB}]
	\Big\}$$
	\sindex{Lambda}{$\Lambda$}%
	\sindex{lambda}{$\lambda$}%
	where it is important that $\xi_{AB}$ and $v^x_R$ are the values in the coordinates of Figure \ref{fig:framerefl} and Figure \myref{fig:fullB};
	clearly	their values are entirely different in any other coordinate system we use.
	$\overline\gamma\in[1,\infty)$ is an arbitrary constant.
	Moreover,
	\begin{alignat}{1}
		\overline\eta^*_C &:= \eta^0_C - \begin{cases}
			0, & \gamma=1, \\
			C_\eta\cdot\epsilon^{1/2}, & \gamma>1,
		\end{cases} \mylabel{eq:Ceta}
	\end{alignat}
	and
	\begin{alignat}{1}
		\overline\xi_{AB} &:= v^x_R - \begin{cases}
			0, & \gamma=1, \\
			C_\xi\cdot\epsilon^{1/2}, & \gamma>1,
		\end{cases} \mylabel{eq:Cxi}
	\end{alignat}
	where $C_\xi,C_\eta$ (to be determined in Proposition \myref{prop:etaa-upperbound}) do not depend on $\epsilon$ or $\lambda$.
	\sindex{Ceta}{$C_\eta$}%
	\sindex{Cxi}{$C_\xi$}%
	$\underline\eta^*_C$ is a constant satisfying $0<\underline\eta^*_C<\overline\eta^*_C$.
	Finally, $\underline\xi_{AB}\in(\xi_{EB},\overline\xi_{AB}]$ may depend on $\gamma$ and $\eta^*_C$.
\end{definition}

\begin{figure}
\input{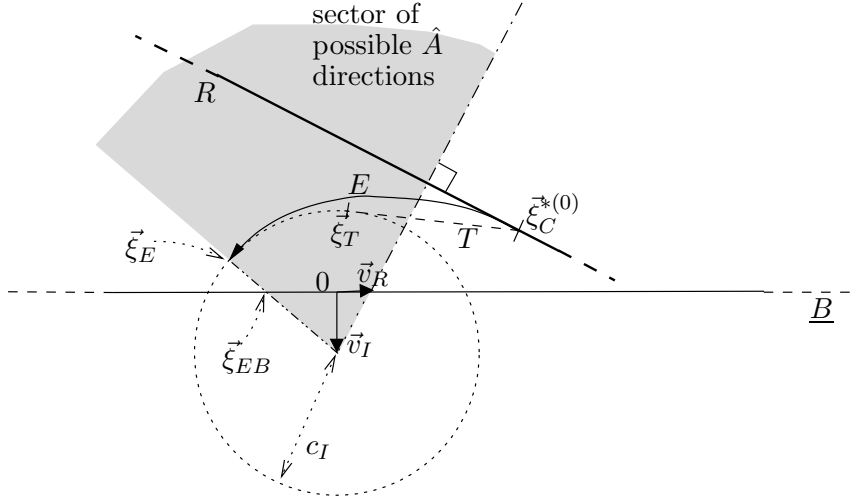}
\caption{The shaded sector consists of all $\hat A$ rays that (1) form a sharp angle with $R$, while (2) meeting $E$ before $E$ meets
$\fullB$ or the circle.}
\mylabel{fig:fullB2}
\end{figure}

\begin{proposition}
	\mylabel{lemma:etax}%
	\mylabel{lemma:Gamma-connected}%
	$\Lambda$ contains $(\gamma,\eta^*_C,\xi_{AB})=(1,\eta^0_C,v^x_R)$ and is path-connected, 
	for $\epsilon$ sufficiently small (depending on $C_\eta,C_\xi$) and $\underline\xi_{AB}$ sufficiently close to $\xi_{EB}$.
\end{proposition}
\begin{proof}
\begin{enumerate}
\item
	We note that the interval $(\xi_{EB},v^x_R]$ has 
	boundaries that are continuous functions of $\lambda$.
\item
	The interval is always nonempty: 
	consider the coordinate system and setting of Figure \myref{fig:fullB}, extended in Figure \myref{fig:fullB2}. 
	Consider a line $T$ through $\vec\xi_C^{*(0)}$ that touches the (upper half of the) circle with center $\vec v_I$ and 
	radius $c_I$ in a point $\vec\xi_T$.
	$T$ can be considered a zero-strength shock (velocity $\vec v_I$, density $\rho_I$ on both sides),
	with $L_d=1$ in $\vec\xi_T$ and $L_d>1$ elsewhere.
	Hence Proposition \myref{prop:shock-envelope} applies: 
	let $\phi\mapsto r(\phi)$ parametrize the line segment from $\vec\xi_C^{*(0)}$ to $\vec\xi_T$; let $\phi\mapsto r_E(\phi)$ parametrize
	$E$. Then $r_E(\phi)>r(\phi)>c_I$ on the interior of the corresponding $\phi$ interval, 
	so the envelope $E$ cannot touch the circle right of $\vec\xi_T$. 
	Moreover, since we have assumed that $M^y_I\leq 1$ (restriction \myeqref{eq:MyI-one}), that means the circle either meets or intersects $\fullB$.
	If $E$ meets $\fullB$ before it meets the circle, then necessarily it meets the part of $\fullB$ \emph{left} of the circle first.

	On the other hand, the extremal choice $\xi_{AB}=v^x_R$
	for $\hat A$ corresponds to (a segment of) the line through $\vec v_I$ and $\vec v_R$ (right side of the shaded sector in Figure
	\ref{fig:fullB2}), which is perpendicular to $R$. Its intersection with the circle is necessarily right of $\vec\xi_T$. 
	Thus: if $E$ meets the circle before it meets $\fullB$, then $\xi_{EB}<v^x_R$ necessarily.
	If $E$ meets $\fullB$ before the circle, then it must meet it left of the origin, so $\xi_{EB}<0<v^x_R$.
	Either way the interval $(\xi_{EB},v^x_R]$ is nonempty.

	Threfore the interval $(\underline\xi_{AB},v^x_R-C_\xi\cdot\epsilon^{1/2}]$ is also nonempty, if $\epsilon$ is sufficiently small 
	(depending on $C_\xi$) and $\underline\xi_{AB}$ sufficiently close to $\xi_{EB}$. 
\item 
	Finally, we show that 
	the special $\lambda=(1,\eta^0_C,v^x_R)$ can be connected by paths in $\Lambda$ to all other $\lambda$:
	it connects to any $(1,\eta^*_C,\xi_{AB})$ with $\eta^*_C\in[\underline\eta^*_C,\eta^0_C)$ and $\xi_{AB}\in[\underline\xi_{AB},v^x_R]$.
	These include $(1,\eta^0_C-C_\eta\cdot\epsilon^{1/2},v^x_R-C_\xi\cdot\epsilon^{1/2})$ which connects to any 
	$(\gamma,\eta^0_C-C_\eta\cdot\epsilon^{1/2},v^x_R-C_\xi\epsilon^{1/2})$ with $\gamma>1$.
	This point, in turn, connects to any $(\gamma,\eta^*_C,\xi_{AB})$ with $\eta^*_C\in[\underline\eta^*_C,\overline\eta^*_C]$
	and $\xi_{AB}\in[\underline\xi_{AB},\overline\xi_{AB}]$.
	Hence $\Lambda$ is path-connected.
\end{enumerate}
\end{proof}

\subsection{Function set and iteration}

\begin{definition}
	\mylabel{def:weighted-hoelder}%
	Let $U\subset\R^n$ open nonempty bounded with $\partial U$ uniformly Lipschitz.
	Let $F\subset\partial U$.
	For $k\in\N_0$, $\alpha\in[0,1]$ and $\beta\in(-\infty,k+\alpha]$ we define
	the \defm{weighted H\"older space} $\spC^{k,\alpha}_\beta(U,F)$ as the set of $u\in\spC^{k,\alpha}(\overline U-F)$ so that 
	$$\|u\|_{\spC^{k,\alpha}_\beta(U,F)}:=
	\sup_{r>0}r^{k+\alpha-\beta}\|u\|_{\spC^{k,\alpha}(\overline U-B_r(F))}$$
	is finite.
\end{definition}

\begin{definition}
	\mylabel{def:b}%
	For sufficiently small $\delta_{\hat P}>0$, 
	there is a function $b\in\spC^2(\overline V)$ 
	with $b,|\nabla b|\leq 1$ so that $b=0$ on $\hat P^{(0)}$, $b>0$ elsewhere,
	$b_n=0$ on $\hat A$ and $\hat B$, and so that $b$ depends continuously on the parameters $\lambda$
	but is independent of $\epsilon$.
	From now on we fix a particular $b$.
	\sindex{b}{$b$}%
\end{definition}
\begin{proof}
	The construction is straightforward. $\delta_{\hat P}$ is taken so small that $\hat P^{(0)}$
	does not meet $\hat A\cup\hat B\cup\{\vec\xi_{AB}\}$ except in $\vec\xi_B^{(0)}$.
\end{proof}
\begin{figure}
	\input{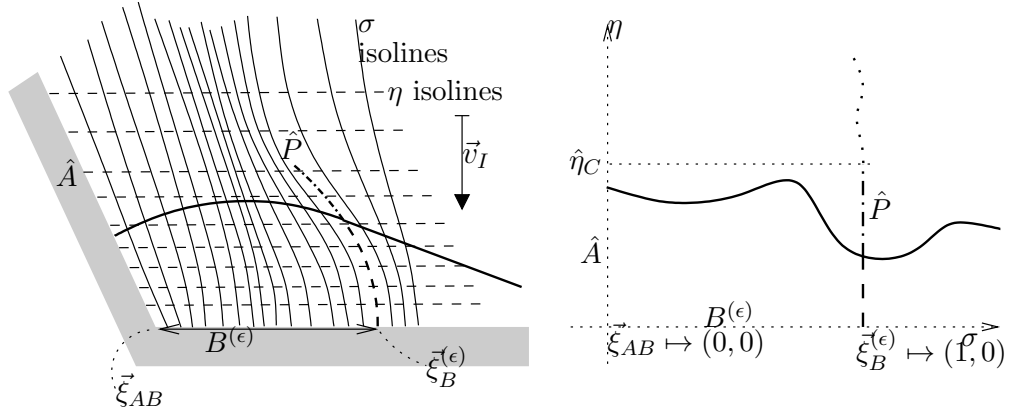}
	\caption{Transformation to ``onion'' coordinates $(\sigma,\eta)$}
	\mylabel{fig:onion}
\end{figure}
\begin{definition}
	\mylabel{def:fusp}%
	\ \par
	{\bf Onion coordinates}

	Rotate Figure \myref{fig:frame0} so that $\hat B$ is the positive horizontal axis (see Figure \myref{fig:onion} left),
	then shift horizontally so that $\vec v_I$ is vertical (Remark \ref{rem:symmetries}).
	Define new coordinates $(\sigma,\eta)\in\R^2$ (see Figure \myref{fig:onion} right) so that
	\begin{enumerate}
	\item
		the coordinate change from $(\xi,\eta)$ to $(\sigma,\eta)$ is $\spC^\infty$ with $\spC^\infty$ inverse,
	\item
		$B^{(\epsilon)}$ maps to $(0,1)\times\{0\}$,
	\item
		$\hat A$ maps to $\{0\}\times(0,\infty)$,
	\item 
		$\hat P^{(\epsilon)}$ maps to $\{1\}\times(0,\hat\eta_C)$ (where $\hat\eta_C$ is the $\eta$ coordinate of the upper endpoint
		of $\hat P^{(\epsilon)}$), 
	\item
		$\vec\xi_B^{(\epsilon)}$ maps to $(1,0)$,
	\item
		$\vec\xi_{AB}$ maps to $(0,0)$.
	\end{enumerate} 
	We require that the change of coordinates and its inverse depend continuously (in the $\spC^\infty$ topology) on $\lambda\in\Lambda$.
	The construction is straightforward.

	Here and in what follows, we will use the weighted H\"older spaces $\spC^{2,\alpha}_\beta(\overline U)$,
	as in Definition \myref{def:weighted-hoelder}.
	The domain $U$ is either $[0,1]^2$ with $F=\{(0,0),(1,1)\}$,
	or $\overline\Omega$ with $F=\{\vec\xi_{AB},\vxia_C\}$ (to be defined).
	For the shock parametrization
	we use $U=[0,1]$ with $F=\{0,1\}$, or (in Figure \myref{fig:onion} left coordinates) $U=[\xi_A,\xi_C]$ with $F=\{\xi_C\}$;
	for $U=P$ we use $F=\{\vec\xi_C\}$, and for $U=A$ or $U=B$ we take $F=\{\vec\xi_{AB}\}$.
	We omit $F$ as it will be clear from the context.
	$\beta\in(1,2)$ and $\alpha\in(0,\beta-1]$ will be determined later.\sindex{alpha}{$\alpha$}\sindex{beta}{$\beta$}
	$C^{2,\alpha}_\beta$ are Banach spaces so that standard functional analysis applies.
	Moreover, 
	$C^{2,\alpha}_\beta(\overline\Omega)$ is continuously embedded in
	$C^1(\overline\Omega)$, so we have $C^1$ regularity in the corners as well, which is crucial.

{\bf Free boundary fit}

	Let $\cfusp$ be the set of functions $\po\in \spC^{2,\alpha}_\beta([0,1]^2)$ that satisfy all of the many conditions explained below.
	Require
	\begin{alignat}{1}
		\|\po\|_{\spC^{2,\alpha}_\beta([0,1]^2)} \leq C_\spC(\epsilon). \myeqlabel{eq:regularity}
	\end{alignat}

	The curves of constant $\sigma$ (isolines) in the $(\xi,\eta)$ coordinate plane are nowhere horizontal, 
	since the other coordinate is $\eta$. Moreover $\psi^I_\eta = v^y_I< 0$ and $\psi^I_\xi=v^x_I=0$, 
	so for all $\sigma\in[0,1]$ there is a unique point $(\xi,s(\sigma))$ on the isoline
	so that 
	\begin{alignat}{1}
		\psi^I(\xi,s(\sigma)) &= \po(\sigma,1). \mylabel{eq:sdef}
	\end{alignat}
	We define another coordinate transform by first mapping $(\sigma,\zeta)\in[0,1]$ to 
	$(\sigma,\eta)$ with $\eta=s(\sigma)\zeta$ and then 
	mapping to $\vec\xi$ with the previous coordinate transform. 

	Let $\vec\xi_A$ resp.\ $\vxia_C$
	\sindex{xiAvector}{$\vxia_A$}%
	\sindex{xiCvector}{$\vxia_C$}%
	be the $\vec\xi$ coordinates for the $(\sigma,\zeta)$ plane points $(0,1)$ and $(1,1)$.
	Let $S$ be the $\vec\xi$ plane curve for $(0,1)\times\{1\}$ (it is the graph of $s$, with endpoints $\vec\xi_A$ and $\vxia_C$). 
	Define $P$ resp.\ $A$ resp.\ $\Omega$\sindex{Omega}{$\Omega$}
	\sindex{P}{$P$}%
	to be the image of $\{1\}\times(0,1)$ resp.\ $\{0\}\times(0,1)$ resp.\ $(0,1)\times(0,1)$.

	Require shock-wall separation:
	\begin{alignat}{1}
		d(S,B) &\geq \delta_{SB} > 0. \myeqlabel{eq:shockwall}
	\end{alignat}
	\sindex{deltaSB}{$\delta_{SB}$}%
	\myeqref{eq:shockwall} ensures that the map from $(\sigma,\zeta)$ to $\vec\xi$ is 
	a well-defined change of coordinates, uniformly nondegenerate (depending on $\delta_{SB}$ and $C_\spC$), 
	with $\spC^{2,\alpha}_\beta([0,1]^2)$ resp.\ $\spC^{2,\alpha}_\beta(\overline\Omega)$ regularity.	
	It is clear now that $\partial\Omega$ is the union of the disjoint sets $S$, $P$, $A$, $B$, and 
	$\{\vxia_C,\vec\xi_B,\vec\xi_A,\vec\xi_{AB}\}$.

	Require: corner close to target: 
	\begin{alignat}{1} 
		|\etaa_C-\etat_C| \leq \epsilon^{1/2}, \myeqlabel{eq:cornerregion}
	\end{alignat}
	We require $\epsilon$ to be so small that $\vxia_C\in\hat P$.
	
	For later use we define $\eta^\pm_C:=\etat_C\pm\epsilon^{1/2}$ and let $\xi^\pm_C$ be so that
	$\vec\xi^\pm_C\in\hat P_C$.
	\sindex{xiplus}{$\xi^+_C$}%
	\sindex{etaplus}{$\eta^+_C$}%
	\sindex{ximinus}{$\xi^-_C$}%
	\sindex{etaminus}{$\eta^-_C$}%

	Corner cone:\sindex{deltaCc}{$\delta_{Cc}$} 
	\begin{alignat}{1}
		\sup_{\vec\xi,\vec\xi'\in\overline\Omega}\measuredangle(\vec\xi-\vec\xi_C,\vec\xi'-\vec\xi_C)
		 &\leq \pi-\delta_{Cc}.
		\myeqlabel{eq:cornercone} 
	\end{alignat}
	($\measuredangle(\vec x,\vec y)$ is the counterclockwise angle from $\vec x$ to $\vec y$.)

{\bf Iteration}

	Here we change to the coordinates of Figure \myref{fig:frame0} for the remainder of the definition.

	Shock strength/density: require that 
	\begin{alignat}{1}
		-\xo-\frac{1}{2}|\nabla\xo|^2 >0, \myeqlabel{eq:rhoprep}
	\end{alignat}
	so that $\rho$ is well-defined (see \myeqref{eq:rhoeq}), and require
	\begin{alignat}{1}
		\min_{\overline\Omega}\rho &\geq \rho_I+\delta_\rho. \myeqlabel{eq:rhomin}
	\end{alignat}
	\sindex{deltarho}{$\delta_\rho$}%

	Pseudo-Mach number bound: require 
	\begin{alignat}{1}
		L^2 \leq 1-\delta_{Lb}\cdot b \qquad\text{in $\overline\Omega$,} \myeqlabel{eq:ellip}
	\end{alignat}
	\sindex{deltaLb}{$\delta_{Lb}$}%
	(Note that $L$ is well-defined because by \myeqref{eq:rhomin} $\rho>0$, so $c>0$.)
	$b=0$ on $\hat P^{(0)}_C$ which has distance $\geq\frac\epsilon3$ (for sufficiently small $\epsilon$) from $\overline\Omega$, so 
	\myeqref{eq:ellip} implies
	\begin{alignat}{1}
		L^2\leq 1-\frac13|\nabla b|_{L^\infty}\delta_{Lb}\cdot\epsilon \leq 1-\frac13\delta_{Lb}\cdot\epsilon
		\qquad\text{in $\overline\Omega$,} \myeqlabel{eq:ellipC}
	\end{alignat}

	Require: there is\footnote{$\pn$ is the product of an iteration step with input $\po$.
		We will ensure in Proposition \myref{prop:pn-uqcont} that $\pn$ is unique and continuously 
		dependent on $\po$.} 
	a function $\pn\in \spC^{2,\alpha}_\beta(\overline\Omega)$ with the following properties:
	\begin{enumerate}
	\item
		$\po$ close to $\pn$: 
		\begin{alignat}{1}
			\|\po-\pn\|_{\spC^{2,\alpha}_\beta([0,1]^2)} &\leq r_I(\po) \myeqlabel{eq:oldnew}
		\end{alignat}
		where $r_I\in C(\cfusp;(0,\infty))$ is a continuous function to be determined later.
		\sindex{rI}{$r_I$}%
	\item
		Right away we require $r_I$ to be so small that
		\begin{alignat}{1}
			-\xn-\frac{1}{2}|\nabla\xn|^2 & >0, \myeqlabel{eq:tilderho} 
		\end{alignat}
		so that in particular $\hat\rho$ is well-defined and positive.
		Moreover, require
		\begin{alignat}{1}
			\nabla\pn &\neq \vec v_I, \myeqlabel{eq:faken}
		\end{alignat}		
	\item
		We require $r_I$ to be so small that (using \myeqref{eq:ellipC})
		\begin{alignat}{1}
			\big(c_0^2+(1-\gamma)(\xo+\frac{1}{2}|\nabla\xn|^2)\big)I-\nabla\xn^2 &>0, \myeqlabel{eq:tildeL}
		\end{alignat}
		i.e.\ is a (symmetric) positive definite matrix.
	\item
		Let $\BL=\BL(\po,\pn)$ be defined in $\vec\xi$ coordinates as
		\sindex{L}{$\BL$}%
		\begin{alignat}{1}
			& \Big(\big(c_0^2+(1-\gamma)(\xo+\frac{1}{2}|\nabla\xn|^2)\big)I-\nabla\xn^2\Big):\nabla^2\pn, \myeqlabel{eq:itn-inner} \\
			& \frac{|\nabla\xn|^2}{2} + \frac{(1-\epsilon)\big((\gamma-1)\xo+c_0^2\big)}{2+(1-\epsilon)(\gamma-1)}, \myeqlabel{eq:itn-parabolic} \\
			& \big(\hat\rho\nabla\xn-\rho_I\nabla\chi^I\big)\cdot\frac{\vec v_I-\nabla\pn}{|\vec v_I-\nabla\pn|}, \myeqlabel{eq:itn-shock} \\
			& \nabla\pn\cdot\vec n_A, \nabla\pn\cdot\vec n_B \Big). \myeqlabel{eq:itn-wall}
		\end{alignat}
		where the codomain is 
		$$Y:=\spC^{0,\alpha}_{\beta-2}(\overline\Omega)
		\times \spC^{1,\alpha}_{\beta-1}(\overline S)
		\times \spC^{1,\alpha}_{\beta-1}(\overline P)
		\times \spC^{1,\alpha}_{\beta-1}(\overline A)
		\times \spC^{1,\alpha}_{\beta-1}(\overline B).$$
		\myeqref{eq:itn-shock} is well-defined by
		\myeqref{eq:tilderho} and \myeqref{eq:faken}. The other components have no singularities.

		Note:
		$\nabla\po\in\spC^{1,\alpha}_{\beta-1}$,
		so $|\nabla\xo|^2\in\spC^{1,\alpha}_{\beta-1}$, 
		so 
		$$\Big(\big(c_0^2+(1-\gamma)(\xo+\frac{1}{2}|\nabla\xn|^2)\big)I-\nabla\xn^2\Big)\in\spC^{1,\alpha}_{\beta-1}
			\contemb\spC^{0,\beta-1}\contemb\spC^{0,\alpha}$$ ($\alpha\leq\beta-1$ as required above), 
		and $\nabla^2\po\in\spC^{0,\alpha}_{\beta-2}$,
		so \myeqref{eq:itn-inner} is $\in\spC^{0,\alpha}_{\beta-2}$.
		In the same way we check that \myeqref{eq:itn-parabolic}, \myeqref{eq:itn-shock} and \myeqref{eq:itn-wall} are $\spC^{1,\alpha}_{\beta-1}$.

		For $\pn$ we use the $C^{2,\alpha}_\beta(\overline\Omega)$ topology.
		We pull back $\pn$ and the value of $\BL$ to $(\sigma,\zeta)$
		coordinates, via the coordinate transform defined by $\po$ (see above), so that we have a fixed domain
		$[0,1]^2$ for all Banach spaces. Then $\BL$ is a nonlinear smooth map
		in the corresponding topologies.

		Most importantly: require 
		\begin{alignat}{1}
			\BL(\po,\pn) &= 0. \myeqlabel{eq:pn}
		\end{alignat}
	\end{enumerate}

{\bf Other bounds}

	Require 
	\begin{alignat}{1}
		\|\po\|_{\spC^{0,1}(\overline\Omega)} \leq C_L \myeqlabel{eq:lip}
	\end{alignat}
	where $C_L$ may not depend on $\epsilon$. 
	\sindex{CL}{$C_L$}%

	$\chi_t$ and $\chi_n$ on parabolic arc: 
	\begin{alignat}{1}
		\max_{\overline P}c^{-1}|\frac{\partial\xo}{\partial t}| &\leq C_{Pt}\cdot\epsilon^{1/2}, \myeqlabel{eq:partan} \\
		\max_{\overline P}c^{-1}\frac{\partial\xo}{\partial n} &\leq -\delta_{Pn}. \myeqlabel{eq:parnor}
	\end{alignat}
	\sindex{CPt}{$C_{Pt}$}%
	\sindex{deltaPn}{$\delta_{Pn}$}%
	We emphasize that $\delta_{Pt},\delta_{Pn}$ may depend \emph{only} on $\lambda$, but not on $\epsilon$ (or $\po$).

	Velocity components:
	\begin{alignat}{1} 
		\vec v\cdot\vec n_A \leq C_{vnA}\cdot\epsilon^{1/2}, 
		\qquad\text{in $\overline\Omega$,} 
		\myeqlabel{eq:Anorvel}
	\end{alignat}
	\begin{alignat}{1} 
		\vec v\cdot\vec t_R \leq \vec v_R\cdot\vec t_R + C_{vtR}\cdot\epsilon^{1/2},
		\qquad\text{in $\overline\Omega$,} 
		\myeqlabel{eq:Rtanvel}
	\end{alignat}
	\begin{alignat}{1}
		\vec v\cdot\vec n_B &\leq \vec v_I\cdot\vec n_B - \delta_{vnB} \qquad\text{in $\overline\Omega$} \myeqlabel{eq:Bnorvel}
	\end{alignat}
	and
	\begin{alignat}{1}
		\vec v\cdot\vec t_A \leq \vec v_I\cdot\vec t_A-\delta_{vtA}
		\qquad\text{in $\overline\Omega$.}
		\myeqlabel{eq:Atanvel}
	\end{alignat}
	\sindex{deltavnB}{$\delta_{vnB}$}%

	Shock normal:
	Let $N\subset S^1$ (unit circle) be the set of $\vec n$ counterclockwise from $\vec n_R$ to $\vec t_A$. Then
	the shock normal satisfies\sindex{deltaSn}{$\delta_{Sn}$}%
	\begin{alignat}{1}
		\sup_Sd(\vec n,N) &\leq C_{Sn}\cdot\epsilon^{1/2}. \myeqlabel{eq:shocknormal}
	\end{alignat}

	Set $\Sigma_1:=A$, $\Sigma_2:=S$, $\Sigma_3:=P$ and $\Sigma_4:=B$. 
	\sindex{Sigma}{$\Sigma_i$}%
	Write the components \myeqref{eq:itn-parabolic}, \myeqref{eq:itn-wall}, \myeqref{eq:itn-shock} of $\BL$
	as 
	$$g^i(\vec\xi,\xn(\vec\xi),\subeq{\nabla\xn(\vec\xi)}{=:\vec p})\qquad (i=1,\dotsc,4),$$ 
	\sindex{gi}{$g^i$}%
	where the $\vec\xi$ dependence includes the dependence on 
	$\xo(\vec\xi)$ and $\nabla\xo(\vec\xi)$. 

	$g^2$ has some singularities, but not on the set of $\vec\xi,\chi,\nabla\chi$ so that \myeqref{eq:Bnorvel}
	and \myeqref{eq:rhomin} (resp.\ \myeqref{eq:tilderho} and \myeqref{eq:faken}) are satisfied. That set is simply connected,
	so we can modify $g^2$ on its complement and extend it smoothly to $\overline\Omega\times\R\times\R^2$.
	The modification is chosen to depend smoothly on $\lambda$. 

	Require uniform obliqueness: 
	\begin{alignat}{1}
		|g^i_{\vec p}\cdot\vec n| &\geq \delta_o|g^i_{\vec p}|\qquad\forall\vec\xi\in\Sigma_i. \myeqlabel{eq:ndobb}
	\end{alignat}
	\sindex{deltao}{$\delta_o$}%

	Functional independence in upper corners: for $i,j=1,4$ and for $i,j=2,3$ set
	$$G:=\begin{bmatrix}g^i_{p^1} & g^j_{p^1}\\g^i_{p^2} & g^j_{p^2}\end{bmatrix},$$
	regard it as a function of $\vec\xi$ (including the dependence on $\nabla\xn(\xi)$) 
	and require\sindex{deltad}{$\delta_d$}\sindex{Cd}{$C_d$}%
	\begin{alignat}{1}
		\|G\|,\|G^{-1}\| \leq C_d \qquad\text{in $B_{\delta_d}(\vec\xi_C)\cap\overline\Omega$.} \myeqlabel{eq:Gb}
	\end{alignat}

	Let $\cfusp$ be the set\footnote{The notation $\cfusp$ does not necessarily imply that $\cfusp$ is the closure of $\fusp$.} 
	of admissible functions so that all of these conditions are satisfied.
	\sindex{F}{$\fusp$, $\fusp_\lambda$, $\cfusp$, $\cfusp_\lambda$}%
	Define
	$\fusp$ to be the set of admissible functions such that all of these conditions are satisfied
	with \emph{strict} inequalities, i.e.\ replace $\leq,\geq$ by $<,>$, ``increasing'' by ``strictly increasing'' etc.

[This is the end of Definition \myref{def:fusp}.]
\end{definition}

The elliptic problem is solved by iteration; $\pn$ is the new iterate, $\po$ the old one. $\BL$ defines $\pn$, as we show later.
As always, the iteration is designed so that its fixed points solve the problem:
\begin{remark}
	\mylabel{rem:fp}%
	If $\pn=\po$, then \myeqref{eq:itn-inner}, \myeqref{eq:itn-shock}, \myeqref{eq:itn-parabolic}, \myeqref{eq:itn-wall}
	and the definition of $S$ yield
	\begin{alignat}{1}
		(c^2I-\nabla\chi^2):\nabla^2\psi &= 0\qquad\text{in $\overline\Omega$,} \notag\\
		\nabla\chi\cdot\vec n &= 0 \qquad\text{on $\overline A$ and $\overline B$,} \notag\\
		\chi^I &= \chi \qquad\text{and} \notag\\
		(\rho\nabla\chi-\rho_I\nabla\chi^I)\cdot\vec n &= 0\qquad\text{on $\overline S$,} \notag\\
		L &= \sqrt{1-\epsilon}\qquad\text{on $\overline P$}\notag
	\end{alignat}
	(we may take closures by regularity \myeqref{eq:regularity}).
\end{remark}

\begin{remark}
\mylabel{rem:reflection}%
Consider a coordinate system where $\vec\xi_{AB}=0$. 
For any point on $A$ or $B$, 
we can use even reflection of $\po$ across the corresponding boundary to obtain a new situation where the point is in the \emph{interior}.
(In $\vec\xi_A$ or $\vec\xi_B$, we obtain a new situation with a point at a shock resp.\ quasi-parabolic arc with an elliptic region on
one side.)
The boundary condition $\chi_n=\po_n=0$ (due to $\vec\xi_{AB}=0$), for even reflection of $\po$, implies that $\po$ is $C^1$ across the 
boundary; then
necessarily it is also $C^{2,\alpha}$.

For fixed points $\po=\pn$, standard regularity theory immediately yields that the solution is locally analytic (even after reflection).
The same technique applied to $\pn$ and to solutions $\pnl$ of linearized equations 
(here $\po$, $\pn$ and $\pnl$ are reflected) yields $C^{2,\alpha}$ regularity. (The same argument applies to $S$ extended by mirror reflection across
$\hat A$.)
\end{remark}

\begin{proposition}
	\mylabel{prop:Lxn-iso}%
	For sufficiently small $\epsilon$ (with bound depending only on $C_{Pt}$) and $r_I$ (depending continuously and only on $\po,\delta_{vx}$):

	for all $\po\in\cfusp$, $\BL(\po,\pn')$ is well-defined for $\pn'$ near $\po$, and the Fr\'echet derivative
	$\partial\BL/\partial\pn'(\po,\po)$ (of $\BL$ with respect to its second argument $\pn'$,
	evaluated at $\pn'=\po$)
	is a linear isomorphism of $\spC^{2,\alpha}_\beta$ onto $Y$. 
\end{proposition}
\begin{proof}
	The proof is almost identical to \pmc{Proposition \pmref{prop:Lxn-iso}}; the new corner between $A,B$ is 
	covered by \cite[Theorem 1.4]{lieberman-crelle-1988} in the same way as the other ones.
\end{proof}

\begin{proposition}
	\mylabel{prop:pn-uqcont}%
	$r_I$ can be chosen so that $\pn$ is unique and depends continuously on $\po\in\cfusp$ (both in the $\spC^{2,\alpha}_\beta$ topology)
	and $\lambda$.
\end{proposition}
\begin{proof}
	The proof is exactly the same as for \pmc{Proposition \pmref{prop:pn-uqcont}}.
\end{proof}

\begin{proposition}
	\mylabel{prop:fusp-topology}%
	For $\epsilon$ and $r_I$ sufficiently small: 
	for all continuous paths $t\in[0,1]\mapsto\lambda(t)$ in $\Lambda$, 
	$\bigcup_{t\in(0,1)}\big(\{t\}\times\fusp_{\lambda(t)}\big)$ is open and
	$\bigcup_{t\in[0,1]}\big(\{t\}\times\cfusp_{\lambda(t)}\big)$ is closed\footnote{We make no statement about $\cfusp$ being the closure
	of $\fusp$. It certainly contains the closure, but it could be bigger, for example if one of the 
	inequalities in Definition \myref{def:fusp} becomes nonstrict in the interior without being violated.}
	in $[0,1]\times C^{2,\alpha}_\beta([0,1]^2)$.
\end{proposition}
\begin{proof}
	All conditions on $\po$ in Definition \myref{def:fusp} are inequalities which can be made scalar by taking a suitable supremum or infimum.
	Then their sides are continuous under $\spC^{2,\alpha}_\beta([0,1]^2)$ 
	changes to $\po$ which, by Proposition \myref{prop:pn-uqcont}, means continuous in $\spC^{2,\alpha}_\beta([0,1]^2)$ change to $\pn$. 
	(Most inequalities need only $\spC^1([0,1]^2)$.) 
\begin{enumerate}
\item Closedness:
	consider sequences $(t_n,\po_n)$ in $\bigcup_{t\in[0,1]}\big(\{t\}\times\cfusp_{\lambda(t)}\big)$ that converge to a limit $(t,\po)$.

	Let $\pn_n$ be associated to $\po_n$ as in Definition \myref{def:fusp}. 
	By continuity (Proposition \myref{prop:pn-uqcont}), $(\pn_n)$
	converges to a limit $\pn$ as well. By continuity of $\BL$ in $\po$, $\pn$ and $\lambda$, we have $\BL_{\lambda(t)}(\po,\pn)=0$ as well.

	Let $\so_n$ be defined by $\po_n$ as in \myeqref{eq:sdef}, with $s\leftarrow s_n$ and $\po\leftarrow\po_n$.
	Then by \myeqref{eq:sdef},
	$(\so_n)$ converges in $\spC^{2,\alpha}_\beta[0,1]$ as well, 
	to a limit $\so$ which satisfies \myeqref{eq:sdef} itself.

	Most conditions on $\po$ are nonstrict inequalities with continuous left- and right-hand side, so they are still satisfied by $\po$.
	We check the strict inequalities explicitly and in order:

	\myeqref{eq:rhoprep} is implied by \myeqref{eq:rhomin}.

	\myeqref{eq:tilderho} resp.\ \myeqref{eq:faken} resp.\ \myeqref{eq:tildeL} are implied by \myeqref{eq:oldnew} resp.\ \myeqref{eq:Bnorvel}
	resp.\ \myeqref{eq:ellipC}, by choosing $r_I$ sufficiently small.

	All inequalities are satisfied, so $\po\in\cfusp$. 

\item Openness:

	same proof, using that all inequalities are strict now, by definition of $\fusp$, hence preserved by sufficiently small perturbations.

\end{enumerate}
\end{proof}

\begin{definition}
	\mylabel{def:it}%
	Define 
	$\IT:\cfusp\rightarrow \spC^{2,\alpha}_\beta([0,1]^2)$
	to map $\po$ into $\pn$ as given in Definition \myref{def:fusp}, but pulled back to $(\sigma,\zeta)$ coordinates
	and the $[0,1]^2$ domain (see Definition \myref{def:fusp}) with the coordinate transform defined by $\po$.
\end{definition}

\subsection{Regularity and compactness}

\begin{proposition}
	\mylabel{prop:regularity}%
	\mylabel{prop:fp-regularity}%
	For sufficiently small $\alpha\in(0,1)$ and $\beta\in(1,2)$, depending only on 
	$C_d$, $\delta_{Lb}\cdot\epsilon$, $\delta_o$, $C_L$, $\delta_{vx}$:
	\begin{enumerate}
	\item
		When parametrized in the coordinates of Figure \myref{fig:framerefl}, 
		\begin{alignat}{1}
			\|S\|_{C^{0,1}} &\leq C_{sL}  \myeqlabel{eq:sLip}
		\end{alignat}
		and
		\begin{alignat}{1}
			\|S\|_{C^{2,\alpha}_\beta} &\leq C_s \myeqlabel{eq:sregu}
		\end{alignat}
		for $C_{sL}=C_{sL}(C_L,\delta_{vx})$ and $C_s=C_s(C_{\spC},\delta_{vx})$;
		the weight $\beta$ is with respect to the endpoints $\vec\xi_A,\vec\xi_C$.
	\item
		For a \emph{fixed point} $\po$ of $\IT$:
		\begin{enumerate}
		\item
			\myeqref{eq:lip} is strict for sufficiently large $C_L$.
		\item 
			\myeqref{eq:regularity} is strict
			for sufficiently large $C_{\spC}=C_{\spC}(C_d,\delta_{Lb}\cdot\epsilon,C_L,\delta_o,\delta_{vtA},\delta_d)$.
		\item
			For $K\Subset\overline\Omega-\hat P-\{\vec\xi_B,\vec\xi_{AB}\}$ and all $k\geq 0$, $\alpha'\in(0,1)$,
			\begin{alignat}{1}
				\|\po\|_{\spC^{k,\alpha'}(K)} &\leq C_{\spC K}
				\myeqlabel{eq:reguint}
			\end{alignat}
			where $C_{\spC K}=C_{\spC K}(d,C_L,\delta_o,\delta_{vtA})$ is decreasing in $d:=d(K,\hat P\cup\{\vec\xi_{AB}\})$ 
			and \emph{not} dependent on $\epsilon$.
		\item 
			$\po$ is analytic in $\overline\Omega-\{\vec\xi_{AB},\vec\xi_C\}$; $S$ is analytic except in $\vec\xi_C$.
		\end{enumerate}
	\item
		For sufficiently small $r_I>0$, depending continuously and only on $\po$, 
		there are $\delta_\alpha,\delta_\beta>0$
		so that for all $\po\in\fusp$, 
		\begin{alignat}{1}
			\|\pn\|_{\spC^{2,\alpha+\delta_\alpha}_{\beta+\delta_\beta}(\overline\Omega)} &\leq C_\IT 
			\myeqlabel{eq:regu2}
		\end{alignat}
		Here, $C_{\IT},\delta_\alpha,\delta_\beta$ depend only on $C_d,\delta_{Lb}\cdot\epsilon,\delta_o,C_L,\delta_{vx}$,
	\end{enumerate}
\end{proposition}
\begin{proof}
	The proof is as the one for \pmc{Proposition \pmref{prop:fp-regularity}}, with obvious modifications. The only additional problem
	is the corner in $\vec\xi_{AB}$. This is very easy to treat with \pmc{Proposition \pmref{prop:corner}} because
	of \myeqref{eq:Gb} and \myeqref{eq:ndobb} for $\vec\xi_{AB}$.
	Note that the corner angle in $\vec\xi_{AB}$ is bounded away from $\pi$ because of the restrictions on $\xi_{AB}$
	(see Section \ref{section:parmset}).
\end{proof}

\begin{remark}
	\myeqref{eq:regu2} implies in particular that $\IT$ is a compact map.
	$\po\in\spC^{2,\alpha}_\beta([0,1]^2)$ is mapped continuously into $\pn\in\spC^{2,\alpha+\delta_\alpha}_{\beta+\delta_\beta}(\overline\Omega)$.
	The latter space is compactly embedded in $\spC^{2,\alpha}_\beta(\overline\Omega)$.
	Pullback to $\spC^{2,\alpha}_\beta([0,1]^2)$ by the $\sigma,\zeta$ coordinates defined by $\po$ (not $\pn$) may destroy the extra
	regularity, but preserves compactness.
\end{remark}

\subsection{Pseudo-Mach number control}

\mylabel{section:L-control}

\begin{proposition}
	\mylabel{prop:Lbounds}%
	For $\epsilon$ and $\delta_{Lb}$ sufficiently small, with bounds depending
	only on $\delta_{\rho}$:
	if $\po\in\cfusp$ is a fixed point of $\IT$, 
	then \myeqref{eq:ellip} is strict and
	\begin{alignat}{1}
		L^2 &< 1-\epsilon \qquad\text{in $\overline\Omega-\overline P$.} \myeqlabel{eq:Leps}
	\end{alignat}
\end{proposition}
\begin{proof}
	$$d(\overline\Omega,\hat P^{(0)})\geq\frac13\cdot\epsilon,$$
	for $\epsilon$ small enough.
	Remember from Definition \myref{def:b} that $b=0$ on $\hat P^{(0)}$. Therefore:
	$$L^2 = 1-\epsilon < 1-\|b\|_{C^{0,1}}\cdot d(P^{(\epsilon)},\hat P^{(0)}) \leq 1-\delta_{Lb}\cdot b\qquad\text{on $\overline P^{(\epsilon)}$,}$$
	e.g.\ for $\delta_{Lb}\leq 1$.

	On the shock, we may use \myeqref{eq:rhomin} combined with \pmc{Proposition \pmref{prop:L-minmax}}
	to rule out that $L^2+\delta_{Lb}\cdot b$ has a maximum in a point where $L<1$ and $L\geq 1-\delta_{LS}$,
	with $\delta_{LS}$ as supplied by loc.cit.
	Here $\delta_{Lb}$ has to be chosen so that $|\delta_{Lb}\nabla b|\leq\delta_{LS}$ is satisfied.
	(Now $\delta_{Lb}$ depends continuously on $\delta_\rho$ as well.)

	In addition we can choose $\delta_{Lb}$ so small that $\delta_{Lb}\cdot b$ satisfies the preconditions of 
	Theorem 1 and Theorem 2 in \cite{elling-liu-ellipticity-journal} (where it is called $b$). 
	For Theorem 2 we use that $b_n=0$ on $\hat A$ and on $\hat B$. 
	Let $\delta_{L\Omega}$ be the $\delta$
	from those theorems (it depends only and continuously on $\lambda$). 
	Then $L^2+\delta_{Lb}\cdot b$ cannot have a maximum in a point of $\Omega\cup A\cup B$ where $L^2\geq 1-\delta_{L\Omega}$.

	In the corner between $A,B$, due to $C^1$ regularity the boundary conditions imply $\nabla\chi=0$, so $L=0$,
	so $L^2+\delta_{Lb}\cdot b=\delta_{Lb}\cdot b<1$ for $\delta_{Lb}$ sufficiently small.

	In $\vec\xi_A$ we use that the shock is pseudo-normal (by the boundary condition $\nabla\chi\cdot\vec n_A=0$ which
	implies $\nabla\chi\cdot\vec t=0$ for the corresponding shock tangent $\vec t$ since $S,A$ form a right angle), so $L_d=L^n_d$
	which is uniformly bounded above away from $1$ by a constant depending on $\delta_\rho$, since \myeqref{eq:rhomin} implies
	uniform shock strength. 

	Assume that \myeqref{eq:ellip} is not strict (or violated). Then $L^2+\delta_{Lb}\cdot b$ has a maximum $\geq 1$ somewhere. For
	$\delta_{Lb}$ sufficiently small (no new dependencies) that means $L^2$ has a maximum $\geq 1-\min\{\delta_{L\Omega},\delta_{LS}\}$ 
	somewhere. But no matter where in $\overline\Omega$ this occurs, it contradicts one of the cases discussed above.
	Hence \myeqref{eq:ellip} is strict.
	
	\myeqref{eq:Leps} can be shown in the same manner, by taking $b=0$ instead, using the actual boundary
	condition $L=\sqrt{1-\epsilon}$ on $P$ and and considering $\epsilon<\delta_{LS},\delta_{L\Omega}$.
\end{proof}

\subsection{Arc control and corner bounds}

\mylabel{section:cornerbounds}

The discussion of parabolic arcs is very similar to \pmc{Sections 4.7 to 4.10}. 
For the convenience of the reader we restate the results using new notation and point out some differences in details.

A new choice of coordinates is convenient (see Figure \myref{fig:frameR}): since self-similar potential flow is
invariant under translations, we may
translate so that $\vec v_R$ moves to the origin (all other velocities $\vec v$
and coordinates $\vec\xi$ have $\vec v_R$ subtracted), then rotate clockwise
until $R$ is horizontal. In this frame, $\vec v_I$ is vertical down and $P$ is centered in $\vec v_R=0$. 
This means $\psi^R$ and $\chi^R$ are both constant on $P$, which simplifies certain calculations. 

\begin{figure}
\input{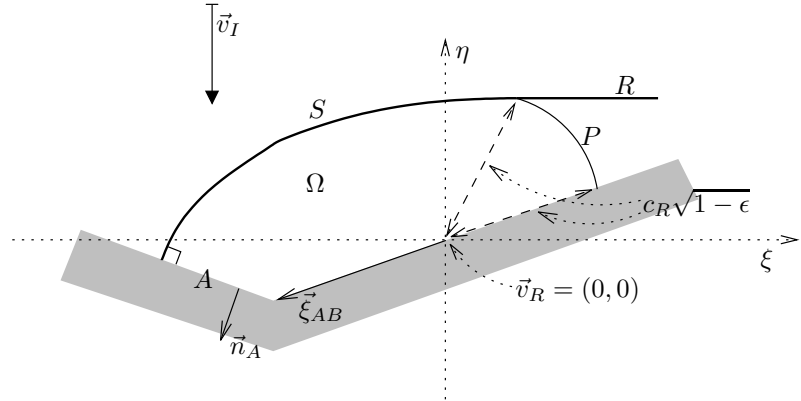}
\caption{In this frame $P$ is centered in $\vec v_R=0$, $R$ is horizontal and $\vec v_I$ is vertical.}
\mylabel{fig:frameR}
\end{figure}

In polar coordinates $(r,\phi)$ with respect to the origin (center of $P$), $P$ corresponds to $r=c_R\cdot\sqrt{1-\epsilon}$.

\begin{proposition}
	\mylabel{prop:pararc}%
	If $C_{Pt}<\infty$ is sufficiently large, if $\delta_{Pn}>0$ is sufficiently small,
	if $\epsilon$ is sufficiently small and $C_{Pv},C_{P\rho}$ sufficiently large,
	with bounds depending only on $C_{Pt}$, 
	then for any fixed point $\chi$ of $\IT$, \myeqref{eq:partan} and \myeqref{eq:parnor}
	are strict, and
	\begin{alignat}{1}
		|\rho-\rho_R| &\leq C_{P\rho}\epsilon^{1/2} \qquad\text{and}  \myeqlabel{eq:rhoP} \\
		|\vec v-\vec v_R| & \leq C_{Pv}\epsilon^{1/2} \qquad\text{on $P$.}  \myeqlabel{eq:vP}
	\end{alignat}
\end{proposition}
\begin{proof}
	The proof is as for \pmc{Proposition \pmref{prop:pararc}}, with obvious modifications.
\end{proof}

If \myeqref{eq:cornerregion} is satisfied, but not in its strict version, then $\eta^*_C=\eta^+_C$ or $\eta^*_C=\eta^-_C$
(where $\vec\xi^\pm$ are as defined in Definition \myref{def:fusp} after \myeqref{eq:cornerregion}). Each of these two cases must be
ruled out.

\begin{proposition}
	\mylabel{prop:etaa-lowerbound}%
	For $\epsilon$ sufficiently small:
		for any fixed point $\po\in\cfusp$ of $\IT$, 
	the lower bound in \myeqref{eq:cornerregion} is strict: 
	$$\etaa_C>\eta_C^-$$
\end{proposition}
\begin{proof}
	Same as for \pmc{Proposition \pmref{prop:etaa-lowerbound}}.
\end{proof}

\begin{proposition}
	\mylabel{prop:a-prop}%
	Consider $\etaa_C=\eta^+_C$.
	For sufficiently small $\epsilon$, there is an $a\geq 0$ so that
	\begin{enumerate}
	\item $\psi+a\xi$ does not have a local minimum (with respect to $\overline\Omega$) at $P\cup\{\vec\xi_B\}$, and 
	\item a shock through $\vec\xi^+_C$ with upstream data $\vec v_I$ and $\rho_I$
		and tangent $(1,\frac{a}{-v^y_I})$ has $v^y_d>0$.
	\end{enumerate}
\end{proposition}
\begin{proof}
	This follows as in Propositions \pmref{prop:cbar}, \pmref{prop:psi-axi} and \pmref{prop:vyd-crit} 
	of \cite{elling-liu-pmeyer-arxiv}.
\end{proof}

Only the final upper bound requires some adaptation:

\begin{proposition}
	\mylabel{prop:etaa-upperbound}%
	Let $\chi\in\overline{\fusp}$ be a fixed point of $\IT$. 
	For $C_\eta$ sufficiently large and for $\epsilon>0$ sufficiently small,
	the upper part of \myeqref{eq:cornerregion} is strict:
	$$\etaa_C < \eta_C^+.$$
\end{proposition}
\begin{proof}
	Again, consider the coordinates of Figure \ref{fig:frameR}. 

	By Proposition \myref{prop:a-prop}, 
	$\psi+a\xi$ cannot have a local minimum at $P\cup\{\vec\xi_B\}$. 
	For $\etaa_C=\eta^+_C$, we have $(\psi+a\xi)_\eta=\psi_\eta>0$ in $\vec\xi_C$ by 
	\pmc{\pmeqref{eq:vyd-eta-positive}} (for sufficiently small $\epsilon$), so the minimum cannot be in $\vec\xi_C$
	either (note that the domain locally contains the ray downward from the corner).

	On the shock (excluding endpoints): let $\xi\mapsto s(\xi)$ be a local parametrization of the shock	.
	$\psi+a\xi=\psi^I+a\xi$, so
	$$\partial_t(\psi+a\xi)=\partial_t(\psi^I+a\xi)=\vec v_I\cdot\vec t+\frac{a}{(1+s_\xi^2)^{1/2}}=\frac{v^y_Is_\xi+a}{(1+s_\xi^2)^{1/2}}.$$ 
	For a local minimum at the shock we need $\partial_t(\psi+a\xi)=0$,
	so
	$$s_\xi=\frac{a}{-v^y_I}.$$
	A \emph{global} minimum, in particular\ $\leq \psi(\vec\xi_C)+a\xi_C$, additionally requires
	that $\vxia_C$ (as well as the rest of the shock) is on or below the tangent through the minimum point, because $\psi^I$ and thus $\psi^I+a\xi$
	are decreasing in $\eta$. 
	By Proposition \myref{prop:a-prop}, the shock through $\vec\xi^+_C$ with that tangent has 
	$v^y_d>0$ for $\etaa_C=\eta^+_C$. In the minimum point the tangent has same slope but is at least as high, 
	so the shock speed is at least as high, so $v^y_d=\psi_\eta=(\psi+a\xi)_\eta$ there is at least as high, in particular $>0$ too.
	But that contradicts a minimum (the ray vertically downwards from any shock point is locally contained in $\overline\Omega$, by
	\myeqref{eq:shocknormal}). 
	Hence $\psi+a\xi$ cannot have a global minimum at the shock.

	The equation \pmeqref{eq:psi} yields $$(c^2I-\nabla\chi^2):\nabla^2(\psi+a\xi)=0$$
	($a\xi$ is linear), so the classical strong maximum principle rules out a minimum in the interior
	(unless $\psi+a\xi$ is constant, which means we are looking at the unperturbed solution which has
	$\etaa_C=\etat_C<\eta^+_C$).

	On $B$, the boundary condition $\psi_n=\chi_n=0$ implies $(\psi+a\xi)_n=a\xi_n\geq 0$ (the slope of $B$ in the frame of Figure \myref{fig:frameR})
	is always nonnegative), so the Hopf lemma rules out a minimum of $\psi+a\xi$ at $B$.

	On $\overline A$ the boundary condition $\chi_n=0$ yields $\psi_n=\vec\xi\cdot\vec n=\vec\xi_{AB}\cdot\vec n_A\geq0$ 
	(see Figure \myref{fig:frameR}). This is actually $\psi_n>0$, except in the special case
	where (in the notation of Definition \myref{def:Lambda}) $\xi_{AB}=v^x_R$
	which is allowed only if $\gamma=1$ and $\eta^*_C=\eta^0_C$: the ``unperturbed'' case.
	In that case, the proof of Proposition \myref{prop:unperturbed-degree-nonzero} shows that only the unperturbed solution
	(Figure \myref{fig:unperturbed}) can solve the problem. Its corner is exactly in the expected location, so that $\eta_C=\eta^*_C<\eta^+_C$.
\end{proof}

\subsection{Velocity and shock normal control}
\mylabel{section:v-control}

\begin{figure}
\input{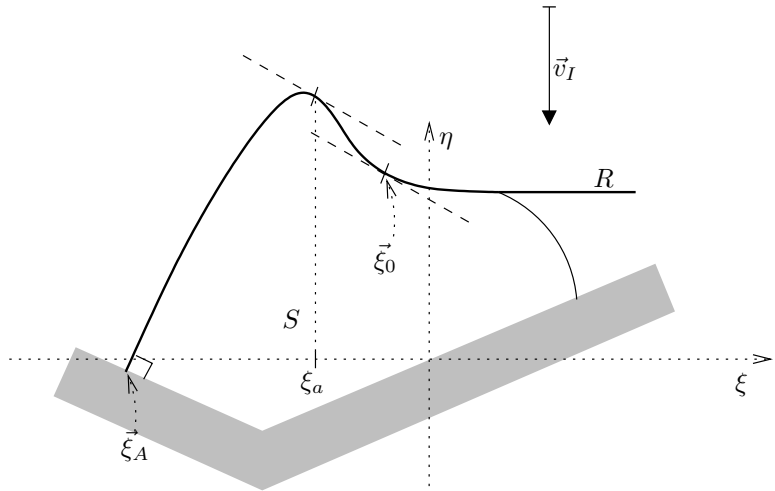}
\caption{A maximum of $v^x$ requires negative curvature, causing a contradiction}
\mylabel{fig:vtR}
\end{figure}

\begin{figure}
\input{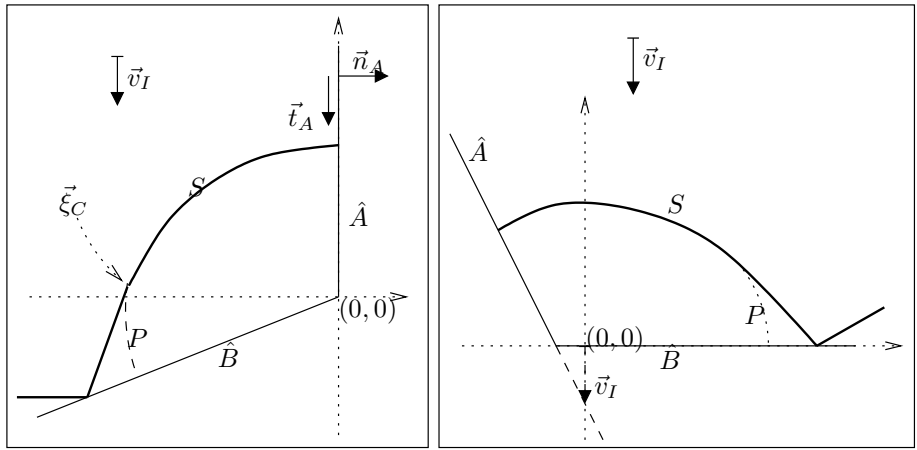}
\caption{Left: mirror-reflect Figure \myref{fig:frame0} across $\hat A$ and rotate around the origin. Right: setting of Figure \ref{fig:onion} left.}
\mylabel{fig:frameAleft}
\mylabel{fig:frameAright}
\end{figure}

\begin{proposition}
	\mylabel{prop:shocktanvel}%
	If $C_{vtR},C_{vnA}$ are sufficiently large (bounds depending only on $C_{Pt}$), 
	if $C_{Sn}$ is sufficiently large (bound depending only on $C_{vtR},C_{vnA}$), 
	if $\epsilon$ is sufficiently small (bound depending only on $C_{Sn}$), and
	if $\delta_{Cc}$ is sufficiently small, then
	for any fixed point $\po\in\cfusp$ of $\IT$,
	the inequalities \myeqref{eq:Rtanvel}, \myeqref{eq:Anorvel},
	\myeqref{eq:shocknormal} and \myeqref{eq:cornercone} are \emph{strict}. Moreover 
	\begin{alignat}{1}
		|\chi_t| &\geq \delta_{\chi t} \qquad\text{on $S\cap B_{\delta_d}(\vec\xi_C)$,} \myeqlabel{eq:chitshock}
	\end{alignat}
	for some constants $\delta_{\chi t},\delta_d>0$.
\end{proposition}
\begin{proof}
\begin{enumerate}
\item 
	For \myeqref{eq:Rtanvel}: consider the coordinates of Figure \myref{fig:frameR} where $\vec t_R=(1,0)$.
	Let $\xi\mapsto s(\xi)$ parametrize $S$ (the shock normal bounds \myeqref{eq:shocknormal} show that $S$ is nowhere vertical in these coordinates,,
	for sufficiently small $\epsilon$, bound depending on $C_{Sn}$). 
	Assume that $\vec v\cdot\vec t_R=v^x$ attains a positive global maximum (with respect to $\overline\Omega$) in a point $\vec\xi_0$ at $S$
	(i.e.\ on the downstream side).
	Since $\vec v_I=(0,v^y_I)$ with $v^y_I<0$, this means $n^x<0$ in $\vec\xi_0$ (because $n^y<0$), i.e.\ $s_\xi(\xi_0)<0$ 
	(see Figure \myref{fig:vtR}). 

	$s_\xi(\xi_0)$ can be expressed as a continuous
	function of $v^x(\vec\xi_0)$ and $\vec\xi_0$. The set of possible $\vec\xi_0$ is contained in the set of possible shock locations which
	is pre-compact. Therefore if $v^x=C_{vtR}\cdot\epsilon^{1/2}$ in $\vec\xi_0\in S$, then 
	\begin{alignat}{1}
		s_\xi(\xi_0) &\leq -C_{s1}\cdot\epsilon^{1/2} \mylabel{eq:s1x0}
	\end{alignat}
	where $C_{s1}=C_{s1}(C_{vtR})>0$ is uniformly increasing in $C_{vtA}$.

	For a constant-state solution \myeqref{eq:Rtanvel} is immediate. Otherwise,
	since $S$ and $\po$ are analytic (Proposition \myref{prop:fp-regularity}), we can apply
	\pmc{Proposition \pmref{prop:vshock}} with $\vec w=(1,0)$, which yields that curvature $\kappa<0$, i.e.\ $s_{\xi\xi}>0$, in $\vec\xi_0$.
	Therefore $s_\xi(\xi)<s_\xi(\xi_0)$ for $\xi<\xi_0$ near $\xi_0$. 
	On the other hand, $s_\xi\geq0$ in $\vec\xi_A$ since the boundary condition $\chi_n=0$ requires the shock to be perpendicular to the wall $A$;
	in particular $s_\xi(\xi_A)>0>s_\xi(\xi_0)$ by \myeqref{eq:s1x0}.
	(In this choice of coordinates, $A$ is either vertical or has negative slope, since we require it to form right or sharp angles with $R$, by
	choice of $\xi_{AB}$ in Section \myref{section:parmset}.)

	Therefore we can pick $\xi_a\in(\xi_A,\xi_0)$ maximal so that $s_\xi(\xi_a)=s_\xi(\xi_0)$.
	Then $s_\xi(\xi)<s_\xi(\xi_0)$
	for $\xi\in(\xi_a,\xi_0)$, so by integration 
	$$s(\xi_a) > s(\xi_0)+s_\xi(\xi_0)\cdot(\xi_a-\xi_0).$$
	But that means the shock tangent in $\xi_a$ is parallel to the one in $\xi_0$
	but \emph{higher}, so the shock speed $\sigma:=\vec\xi\cdot\vec n$ is smaller. 
	By \pmc{\pmeqref{eq:DvndDsigma}}, that means $v^n_d$ is smaller, whereas
	$v^t$ is the same (parallel tangents). $n^x<0$, so $v^x_d$ is \emph{bigger}. Contradiction --- we assumed that
	we have a \emph{global} maximum of $v^x$ in $\vec\xi_0$.

	\pmc{Propositions \pmref{prop:interior-velocity} and \pmref{prop:v-wall}}
	rule out local maxima of $v^x$ in $\Omega$ and at $B$,
	where we use that $\chi$ is analytic and that $(1,0)$ is not vertical, i.e.\ not normal to $B$.
	
	At $A$: if $A$ is vertical, then the boundary condition requires $v^x=\xi_A<0$; if $A$ is not vertical, then $(1,0)$ is not normal,
	so \pmc{Proposition \pmref{prop:interior-velocity}} applies again.

	In $\vec\xi_{AB}$, the two boundary conditions combine to yield $\vec v=\vec\xi_{AB}$, so $v^x=\xi_{AB}<0$.
	
	In $\vec\xi_A$, $s_\xi\geq0$ (see above) yields $v^x\leq 0$. 

	On $\overline{P}$ we can use \myeqref{eq:vP} with $v^x_R=0$, increasing $C_{vtR}$ to $>C_{Pv}$ if necessary 
	(this makes $C_{vtR}$ depend on $C_{Pt}$ as well). 

	All parts of $\overline\Omega$ are covered; \myeqref{eq:Rtanvel} is strict.
\item
	For \myeqref{eq:Anorvel}: consider the coordinates of Figure \myref{fig:frameAleft} left.
	There, $\vec n_A=(1,0)$, so we need to show $v^x=\vec v\cdot\vec n_A\leq C_{vnA}\cdot\epsilon^{1/2}$.
	On $\overline{A}$, the boundary condition yields $v^x=0$. $B$ is never vertical, so \pmc{Proposition \pmref{prop:v-wall}} rules out extrema
	of $v^x$ at $B$. \pmc{Proposition \pmref{prop:interior-velocity}} does not allow extrema in $\Omega$. At $P$,
	\myeqref{eq:vP} yields $v^x=v^x_R+O(\epsilon^{1/2})$; note that $v^x_R<0$ in these coordinates.
	At $S$, we can use the same curvature 
	argument as for $\vec v\cdot\vec t_R$, except that we now use $s_\xi\geq 0$ in $\vec\xi_C$ rather than $\vec\xi_A$.
	Altogether we obtain a contradiction again, if $C_{vnA}$ is sufficiently large, depending only and continuously on $C_{Pt}$.
\item 
	Consider the coordinates of Figure \myref{fig:frameR}.
	The slope $s_\xi$ of some shock passing through a point $\vec\xi$ is uniquely determined by (and continuous in) $\vec\xi$ and $v^x$, 
	with $\sign s_\xi=-\sign v^x$ (since $\vec v_I=(0,v^y_I)$, $v^y_I<0$). 
	The set of possible shock locations $\vec\xi$ is pre-compact, so \myeqref{eq:Rtanvel} implies
	$$\sup\measuredangle(\vec n,\vec n_R)<C_{Sn}\cdot\epsilon^{1/2}$$
	where $C_{Sn}=C_{Sn}(C_{vtR})$. 	

	Analogously we argue that \myeqref{eq:Atanvel} implies
	$$\sup\measuredangle(\vec t_A,\vec n)<C_{Sn}\cdot\epsilon^{1/2},$$
	where $C_{Sn}=C_{Sn}(C_{vtR},C_{vnA})$ now.
	\myeqref{eq:shocknormal} is strict with these choices.
\item
	These shock normal bounds also imply \myeqref{eq:cornercone} is strict, for $\delta_{Cc}>0$ and $\epsilon>0$ sufficiently small(er),
	with $\epsilon$ bound depending only on $C_{Sn}$.
\item
	Near each corner the shock normal bound bounds $\vec n$ away from the $\vec\xi$ direction, so 
	$|\chi^I_t|\geq\delta_{\chi t}$ and therefore \myeqref{eq:chitshock} for some $\delta_{\chi t}$. 
\end{enumerate}
\end{proof}

\mylabel{section:densitycontrol}

\begin{proposition}
	\mylabel{prop:shockenv}%
	\begin{enumerate}
	\item 
		If $\delta_{SB}$ is sufficiently small, then \myeqref{eq:shockwall} is strict.
	\item There is a constant $\delta_{\rho S}>0$ so that 
		\begin{alignat}{1}
			\rho_d &\geq \rho_I+\delta_{\rho S} \qquad\text{at $\overline S$}\myeqlabel{eq:shockrho}
		\end{alignat}
	\end{enumerate}
\end{proposition}
\begin{proof}
\begin{enumerate}
\item
	Consider the envelope $E$ defined in Section \myref{section:parmset}. 
	The parameter set $\Lambda$ (see Definition \myref{def:Lambda}) has been chosen so that for any $\lambda\in\Lambda$, 
	$E$ passes from $\vec\xi^{*(0)}_C$ to $\hat A$ without meeting $\hat B$ or 
	the circle (with radius $c_I$ centered in $\vec v_I$). Since $\Lambda$ has also been chosen compact, 
	$E$ is in fact uniformly bounded away from $\hat B$ and the circle. 

	$E$ starts in $\vec\xi_C^{*(0)}$; 
	let $E'$ be the counterclockwise envelope (Definition \myref{def:envelope}) starting in $\vec\xi_C$ instead. 
	$E,E'$ are solutions of an ODE \myeqref{eq:envelope-ode}, so they depend continuously on the initial point.
	Hence for $\vec\xi_C$ sufficiently close to $\vec\xi_C^{*(0)}$, i.e.\ by \myeqref{eq:cornerregion} for sufficiently small $\epsilon$ (with upper bound depending
	only on the choice of $\Lambda$), 
	$E'$ is also uniformly bounded away from $\hat B$ and the circle.

	Now we can apply the argument displayed in Figure \ref{fig:theorem} right:
	$|\vec\xi-\vec v_I|$ is $r$ in the polar coordinates used in Section \myref{section:envelope}.
	Let $E'$ and the shock $S$ be parametrized by $\phi\mapsto r_S(\phi)$ resp.\ $\phi\mapsto r_{E'}(\phi)$, with
	$\phi\in[\phi_C,\phi_A]$, $\phi_C$ corresponding to the ray from $\vec v_I$ through $\vec\xi_C$ and $\phi_A$ to the ray from $\vec v_i$
	containing $\hat A$.
	$r_S(\phi_C)=r_{E'}(\phi_C)$ because $S$ and $E'$ both pass through $\vec\xi_C$. 
	By \myeqref{eq:ellipC}, $L_d<1$ at $S$. 
	Therefore, Proposition \myref{prop:shock-envelope} yields $r_S(\phi)>r_{E'}(\phi)$ for all $\phi>\phi_C$.
	Hence topologically $S$ is separated from $\hat B$ and the circle by $E'$, so it also has uniformly lower bounded distance from
	them. In particular \myeqref{eq:shockwall} is strict, for sufficiently small $\delta_{SB}$ (depending only on the choice of $\Lambda$, 
	but not on any other constant).
\item
	If $S$ vanishes in some point $\vec\xi$, then $L_d=L_u=|\vec\xi-\vec v_I|/c_I$ which --- since $S$ has uniform distance
	from the circle --- is uniformly bounded below away from $1$.
	However, this contradicts \myeqref{eq:ellipC}. The shock cannot vanish; on the contrary, by continuity the shock has uniformly 
	lower-bounded strength.
	That implies \myeqref{eq:shockrho}, for sufficiently small $\delta_{\rho S}$. (Again, it depends only on $\Lambda$, not on
	the choice of other constants.)
\end{enumerate}
\end{proof}

\begin{proposition}
	\mylabel{prop:rho}%
	If $\delta_\rho$ and $\epsilon$ are sufficiently small (with bounds depending only on $C_{Pt}$),
	then for any fixed point $\po\in\cfusp$ of $\IT$,
	the inequality \myeqref{eq:rhomin} is \emph{strict}.
\end{proposition}
\begin{proof}
	By Proposition \myref{prop:fp-regularity}, $\psi$ and hence $s$ are analytic. Thus we may use \pmc{Proposition \pmref{prop:c-principle}}
	which rules out minima of $\rho$ in $\Omega$ and (using Remark \myref{rem:reflection}) at $A$ or $B$. 

	Consider the coordinates of Figure \myref{fig:frame0}.
	In $\vec\xi_A$, the first shock condition is 
	$$\psi(\vec\xi_A)=\psi^I(\vec\xi_A)=-\pi(\rho_I)+v^x_I\big(\xi_A-\frac12v^x_I\big).$$
	\myeqref{eq:Atanvel} implies
	$$\psi(\vec\xi_{AB})\leq\psi(\vec\xi_A)+(\subeq{\xi_{AB}}{=0}-\subeq{\xi_A}{<0})(v^x_I-\delta_{vtA})
	=-\pi(\rho_I)+\subeq{\delta_{vtA}\xi_A-\frac12(v^x_I)^2}{<0}.$$
	So in $\vec\xi_{AB}=0$, since $\nabla\chi=0$ by boundary conditions on $A,B$ and $C^1$ regularity:
	$$\rho=\pi^{-1}(-\chi-\frac12|\nabla\chi|^2)=\pi^{-1}(-\psi)=\rho_I+\delta_{\rho AB}$$
	for some constant $\delta_{\rho AB}>0$ depending only on the parameters $\lambda$; note that $\pi$ is a strictly increasing function
	for any $\gamma\geq 1$.
	We can pick $\delta_\rho<\delta_{\rho AB}$ so that $\rho\leq\rho_I+\delta_\rho$ is not possible in $\vec\xi_{AB}$.

	On $P$ we know $\rho$ up to a small constant, by \myeqref{eq:rhoP}, so we can choose $\delta_\rho$ even smaller
	so that $\rho\leq\rho_I+\delta_\rho$ is not possible at $\overline P$.

	By \myeqref{eq:shockrho}, $\rho$ at $\overline S$ is uniformly bounded below away from $\rho_I$. 
	Hence, for $\delta_\rho$ sufficiently small, $\rho$ cannot have a global minimum close to $\rho_I$ at $S$.

	We see that for sufficiently small $\delta_\rho$ and $\epsilon$, depending
	continuously on $C_{Pt}$ (and $\lambda$), 
	\myeqref{eq:rhomin} is \emph{strict}.
\end{proof}

\begin{proposition}
	\mylabel{prop:shocknorvel}%
	If $\delta_{vtA}$, $\delta_{vnB}$ and $\epsilon$ are sufficiently small ($\delta_{vtA}$, $\delta_{vnB}$ bounds depending only on 
	$\delta_\rho,C_{Sn}$, $\epsilon$ bound depending only on $C_{Pt}$), 
	and if $\delta_{Cc}$ is sufficiently small,
	then for any fixed point $\po\in\cfusp$ of $\IT$,
	the inequalities \myeqref{eq:Atanvel} and \myeqref{eq:Bnorvel}
	are \emph{strict}.
\end{proposition}
\begin{proof}
	Consider the coordinates of Figure \myref{fig:frameAright} right, where
	$\vec v\cdot\vec n_B=-v^y$. 
	\myeqref{eq:rhomin} implies that the shock is uniformly strong.
	By \myeqref{eq:shocknormal}, the shock normal $\vec n$ is everywhere downwards and uniformly not horizontal.
	Thus $v^y>v^y_I+\delta_{vnB}$ at $\overline S$ for sufficiently small $\delta_{vnB}$, depending only on $\delta_\rho$ and $C_{Sn}$.

	\pmc{Proposition \pmref{prop:interior-velocity}} rules out local maxima of $v^y$ in $\Omega$.

	If $v^y$ has a local maximum at $A$, then $A$ must be horizontal (\pmc{Proposition \pmref{prop:v-wall}}), but by construction
	it is not. 

	On $\overline B$ the boundary condition implies $0=\chi_n=\chi_2$, so
	$v^y=\psi_2=\eta_{AB}=0$.
	
	At $\overline{P}$ we can use \myeqref{eq:vP} with $v^y_R=0>v^y_I$ to obtain $v^y>v^y_I$ if $\epsilon$ is small enough (depending on $C_{Pt}$).
	
	Altogether we have that \myeqref{eq:Bnorvel} is strict if $\delta_{vnB}$ is small enough.

	The arguments for \myeqref{eq:Atanvel} are analogous, looking at Figure \myref{fig:frameAleft} left coordinates instead:
	the shock $S$ is nowhere vertical (by \myeqref{eq:shocknormal}), so $v^y>v^y_I+\delta_{vtA}$ at $\overline S$
	for sufficiently small $\delta_{vtA}$. 
	If $B$ is not horizontal, then the direction $(0,1)$ is not perpendicular to it, so \pmc{Proposition \pmref{prop:v-wall}} 
	rules out a local $v^y$ extremum at $B$; if $B$ is horizontal, then $0=\chi_n=\chi_2$, so $v^y=\psi_2=\eta_{AB}=0$ on it.
	$A$ is always vertical, i.e.\ never perpendicular to $(0,1)$, so by \pmc{Proposition \pmref{prop:v-wall}} no $v^y$ extremum is possible at it.
	In $\vec\xi_{AB}=0$, the boundary conditions combine to $\vec v=0$, so $v^y=0>v^y_I+\delta_{vtA}$ if $\delta_{vtA}$ is small enough.
	At $\overline P$ we can use \myeqref{eq:vP} again to obtain $v^y\geq v^y_R-C_{Pt}\cdot\epsilon^{1/2}>v^y_I$ (using $v^y_R>v^y_I$ and 
	for $\epsilon$ sufficiently small,
	with bound depending only on $C_{Pt}$). \pmc{Proposition \pmref{prop:interior-velocity}} rules out interior extrema of $v^y$.
	Hence \myeqref{eq:Atanvel} is strict if $\delta_{vtA}$ is small enough.
\end{proof}

\subsection{Fixed points}

\begin{proposition}
	\mylabel{prop:oblique-corner}%
	For $\delta_o$ sufficiently small, with bounds depending only on $\delta_\rho$ and $C_L$,
	for $C_d$ resp.\ $\delta_d$ sufficiently large resp.\ small, with bounds depending only on $\delta_\rho$ and $C_L$,
	and for $\epsilon$ sufficiently small, with bounds depending only on $C_{Pt}$, $C_L$ and $\delta_\rho$: 

	If $\xo\in\cfusp$ is a fixed point of $\IT$, then \myeqref{eq:ndobb} and \myeqref{eq:Gb} are strict.
\end{proposition}
\begin{proof}
	Compared to \pmc{Proposition \pmref{prop:oblique-corner}}, the only new case is a corner between two walls, $A$ and $B$.
	The corner angle is bounded away from $0$ and $\pi$ by constants depending only on the parameters $\lambda$.
	(Note that $\xi_{AB}$ in Section \myref{section:parmset} has been lower-bounded uniformly by $\underline\xi_{AB}$ in Definition
	\myref{def:Lambda}, so that $\hat A,\hat B$ are uniformly not parallel.)
	$g_{\vec p}$ on $A$ and $B$ is their respective normal, so \myeqref{eq:Gb} is obvious. 
\end{proof}

\begin{proposition}
	\mylabel{prop:fp-boundary}%
	If the constants in \myeqref{eq:constlist} in Definition \myref{def:fusp} are chosen sufficiently small resp.\ large:
	for any $\lambda\in\Lambda$, 
	$\IT_{\lambda}$ cannot have fixed points on $\cfusp_{\lambda}-\fusp_{\lambda}$.
\end{proposition}
\begin{proof}
	Let $\xo\in\cfusp$ be a fixed point of $\IT$.
	We show that every inequality in the definition of $\cfusp$ is strict,
	so $\xo\in\fusp$.

	\myeqref{eq:regularity} and \myeqref{eq:lip} are strict by Proposition \myref{prop:regularity}.

	\myeqref{eq:shockwall} is strict by Proposition \myref{prop:shockenv}.

	\myeqref{eq:rhomin} is strict by Proposition \myref{prop:rho}.

	A fixed point satisfies $\po=\pn$, so $\|\po-\pn\|=r_I(\po)>0$ cannot be true.
	\myeqref{eq:oldnew} is strict.

	\myeqref{eq:ellip} strict is provided by Proposition \myref{prop:Lbounds}.

	Due to Proposition \myref{prop:Lbounds}, $L^2=1-\epsilon$ on each point of $\overline{P}$,
	so we are in the situation of Section \myref{section:cornerbounds} and \pmc{Section \pmref{section:parcs} etc}.
	Proposition \myref{prop:pararc} shows that \myeqref{eq:partan} and \myeqref{eq:parnor} are strict.

	\myeqref{eq:Anorvel} and \myeqref{eq:Rtanvel} are strict by Proposition \myref{prop:shocktanvel}.

	\myeqref{eq:Bnorvel} and \myeqref{eq:Atanvel} are strict by Proposition \myref{prop:shocknorvel}.

	Propositions \myref{prop:etaa-lowerbound} and \myref{prop:etaa-upperbound} rule out
	$\etaa_C=\etat_C\pm\delta^{-1}\epsilon$ if $\delta$ is small enough,
	so \myeqref{eq:cornerregion} is strict.

	\myeqref{eq:cornercone} is strict by Proposition \myref{prop:shocktanvel}.

	\myeqref{eq:regularity} yields a trivial upper bound on the density in $\overline\Omega$, hence downstream at the shock.

	\myeqref{eq:shocknormal} is strict by Proposition \myref{prop:shocktanvel}.

	Proposition \myref{prop:oblique-corner} shows that \myeqref{eq:ndobb} and \myeqref{eq:Gb} are strict.

	All inequalities are strict, so $\po\in\fusp$.
\end{proof}

\subsection{Existence of fixed points}

\mylabel{section:ls}

\begin{figure}
\input{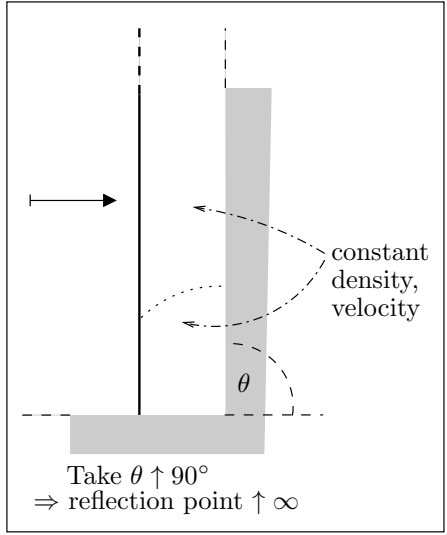}
\caption{The unperturbed case: a straight vertical shock $R$. In this case there is no reflection point and no incident shock.}
\mylabel{fig:unperturbed}
\end{figure}

We determine the Leray-Schauder degree of $\IT$ on $\fusp$ for a particular choice of parameters $\lambda$: the unperturbed problem 
(see Figure \myref{fig:unperturbed}),
featuring a straight shock separating two constant-state regions
($\etat_C=\overline\eta^0_C$, $\xi_{AB}=v^x_R$ in the coordinates of Definition \myref{def:Lambda}),
for $\gamma=1$. 

\begin{proposition}
        \mylabel{prop:unperturbed-degree-nonzero}%
	For sufficiently small $\epsilon$:

        For $\gamma=1$, $\eta^*_C=\overline\eta^*_C$ and $\xi_{AB}=v^x_R$, $\IT$ has nonzero Leray-Schauder degree.
\end{proposition}
\begin{proof}
        We can use reflection across $A$ (Remark \ref{rem:reflection}) to obtain the problem
	of Propositions 4.14.1 and 4.14.3 in \cite{elling-liu-pmeyer-arxiv}. The resulting iteration $\IT$ is almost
	the same as in loc.cit., except for minor differences in the coordinate transform from $(\sigma,\zeta)\in[0,1]^2$ (fixed domain)
	to $\vec\xi$ coordinates
	(see Definition \myref{def:fusp} as compared to \pmc{Definition \pmref{def:fusp}}). The proofs of \pmc{Propositions
	\pmref{prop:unperturbed-unique} and \pmref{prop:unperturbed-index}} carry over without any change
	to show that the present problem has nonzero Leray-Schauder degree.
\end{proof}

\begin{proposition}
	\mylabel{prop:probell}%
	For sufficiently small resp.\ large constants in \myeqref{eq:constlist}:
	$\IT$ has a fixed point for all $\lambda\in\Lambda$.
\end{proposition}
\begin{proof}
	The proof is identical to the one of \pmc{Proposition \pmref{prop:probell}}, 
	except for the definition of $\Lambda$ (Definition \myref{def:Lambda}); 
	we use the known Leray-Schauder degree in 
	$(\gamma,\eta^*_C,\xi_{AB})=(1,\eta^0_C,v^x_R)$ from Proposition \myref{prop:unperturbed-degree-nonzero}.
\end{proof}

\subsection{Construction of the entire flow}
\mylabel{section:entireflow}

\begin{figure}
\input{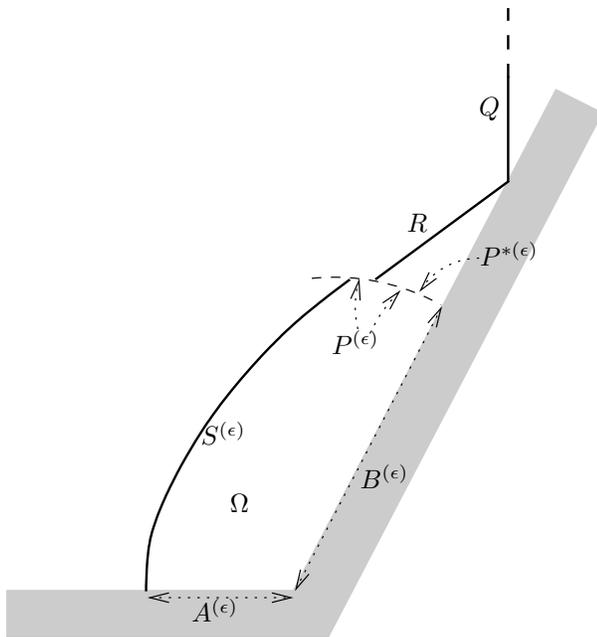}
\caption{The expected and actual parabolic arc ($P^{*(\epsilon)}$ and $P^{(\epsilon)}$) differ by curve of length $O(\epsilon^{1/2})$ 
(by \myeqref{eq:cornerregion})}
\mylabel{fig:flowmatch}
\end{figure}

\begin{proof}[Proof of Theorem \myref{th:elling-rrefl}]
	For all $\rho_I,c_I,M_I\in(0,\infty)$ and
	for each choice (in Definition \myref{def:Lambda}) of $\overline\gamma$, $\underline\eta^*_C$ and $\underline\xi_{AB}$
	we obtain a separate parameter set $\Lambda$.
	For sufficiently small constants in \myeqref{eq:constlist}, Proposition \myref{prop:probell} 
	yields fixed points $\po$ for all $\lambda\in\Lambda$. 
	Note that there is no lower bound on $\epsilon$, except that $\alpha,\beta$ etc.\ may change as $\epsilon\downarrow 0$.

	By Definition \myref{def:fusp}, Remark \myref{rem:fp}, 
	Proposition \myref{prop:pararc} and \myeqref{eq:reguint}, the fixed points satisfy
	\begin{alignat}{3}
		(c^2I-\nabla\chi^2):\nabla^2\psi &= 0 && \qquad \text{in $\Omega^{(\epsilon)}$}, \myeqlabel{eq:interior-eps} \\
		|\psi-\psi^R(\vxit_C)| & = O(\epsilon^{1/2}) && \qquad \text{and} \myeqlabel{eq:para-chi-eps} \\
		|\rho-\rho_R| & = O(\epsilon^{1/2}) && \qquad\text{and} \myeqlabel{eq:para-rho-eps} \\
		|\nabla\psi-\vec v_R| & = O(\epsilon^{1/2}) && \qquad \text{on $P^{(\epsilon)}$,} \myeqlabel{eq:para-nablachi-eps} \\
		\chi &= \chi^I && \qquad \text{and} \myeqlabel{eq:shock1-eps}  \\
		(\rho\nabla\chi-\rho_I\nabla\chi^I)\cdot\vec n &= 0 && \qquad \text{on $S$,} \myeqlabel{eq:shock2-eps}  \\
		\nabla\chi\cdot\vec n &= 0 && \qquad \text{on $A\cup B$,} \myeqlabel{eq:wall-eps}\\
		|\vxia_C-\vec\xi^{*(\epsilon)}_C| &= O(\epsilon^{1/2}) \myeqlabel{eq:cornerdist-eps}
	\end{alignat}
	where the $O$ constants are independent of $\epsilon$. For regularity, Proposition \myref{prop:regularity} yields
	\begin{alignat}{1}
		\|\psi\|_{C^{0,1}(\overline\Omega^{(\epsilon)})} &\leq C_1, \myeqlabel{eq:lip-eps} \\
		\|\psi\|_{C^{k,\alpha}(K\cap\overline\Omega^{(\epsilon)})},|S|_{C^{k,\alpha}(K\cap\overline S^{(\epsilon)})} &\leq C_2(d) \myeqlabel{eq:cka-eps}\\
		\qquad\text{where $d:=d(K,\hat P^{(\epsilon)}\cup\{\vec\xi_{AB}\})>0$} \notag.
	\end{alignat}
	for constants $C_1$ and $C_2(d)$ independent of $\epsilon$.

	Now consider those parameter vectors $\lambda$ that arise from the situtation in Theorem \ref{th:elling-rrefl}, i.e.\ so that
	there is an incident shock $Q$ meeting $R$ in a local regular reflection.
	We extend $\po$ from above to a function $\psi^{(\epsilon)}$ 
	defined on all of $\overline V$ as shown in Figure \myref{fig:flowmatch}:
	set $\rho=\rho_R$, $\vec v=\vec v_R$ in the region enclosed by $R$ shock, $\hat B$ and $P^{*(\epsilon)}$;
	set $\rho=\rho_Q$, $\vec v=\vec v_{R,Q}$ in the region right of the $Q$ shock and
	$\rho=\rho_I$, $\vec v=\vec v_I$ in the remaining area. 
	In each of the four regions, $\psi^{(\epsilon)}$ is a strong solution of self-similar potential flow,
	so we can multiply the divergence-form PDE \pmc{\pmeqref{eq:chi-divform}} with any test function $\vartheta\in C_c^\infty(\overline V)$
	and integrate over all region to obtain a sum of boundary integrals of the type
	$$\int_M\rho\nabla\chi\cdot\vec n~ds$$
	where $M$ are various curves; $\nabla\chi$ and $\rho$ are limits on one of the sides of $M$.

	The symmetric difference of $P^{(\epsilon)}$ and
	$P^{*(\epsilon)}$ has length $O(\epsilon^{1/2})$ (by \myeqref{eq:cornerdist-eps}, so since $\nabla\psi$ and $\psi$ are bounded in each region
	(uniformly in $\epsilon$, by \myeqref{eq:lip-eps}), the boundary integral over the difference contributes only $O(\epsilon^{1/2})$. 
	The difference of the integrals on each side of $P^{*(\epsilon)}\cap P^{(\epsilon)}$ are $O(\epsilon^{1/2})$ due to
	\myeqref{eq:para-rho-eps} and \myeqref{eq:para-nablachi-eps}. The integrals over $A,B$ vanish due to \myeqref{eq:wall-eps}.
	Finally, the integrals on each side of $S^{(\epsilon)}$ cancel due to \myeqref{eq:shock1-eps} and \myeqref{eq:shock2-eps}. 
	Altogether:
	\begin{alignat}{1}
		\int_{\overline V}\rho^{(\epsilon)}\nabla\chi^{(\epsilon)}\cdot\nabla\vartheta-2\rho^{(\epsilon)}\vartheta~d\vec\xi = O(\epsilon^{1/2}).
		\mylabel{eq:eps-weak}
	\end{alignat}

	$\spC^{k,\alpha}$ with $k+\alpha>1$ is compactly embedded in $C^{0,1}$, so by \myeqref{eq:cka-eps} with a diagonalization argument, 
	for every compact $K\subset\overline V-\{\vec\xi_{AB}\}-\overline P^{*(0)}$
	we can find a sequence $(\epsilon_k)\downarrow 0$ so that $\psi^{(\epsilon_k)}$
	converges to $\psi^{(0)}$ in $C^{0,1}(K)$. 
	Moreover $\rho^{(\epsilon)}$ and $\nabla\chi^{(\epsilon)}$ are bounded on $\overline V$ uniformly in $\epsilon$,
	so we may take $\epsilon\downarrow 0$ in \myeqref{eq:eps-weak} to obtain
	\begin{alignat}{1}
		\int_V \rho^{(0)}\nabla\chi^{(0)}\cdot\nabla\vartheta-2\rho^{(0)}\vartheta~d\vec\xi &= 0. \myeqlabel{eq:zero-weak}
	\end{alignat}
	In addition, \myeqref{eq:shock1-eps} and \myeqref{eq:para-chi-eps} combined with \myeqref{eq:lip-eps} show that 
	\begin{alignat}{1}
		\psi^{(0)} &\in C(\overline V) \myeqlabel{eq:zero-cont}
	\end{alignat}
	Finally, by construction of $\psi^{(\epsilon)}$, 
	\begin{alignat}{1}
		\rho^{(0)}(s\vec\xi),\vec v^{(0)}(s\vec\xi) &\rightarrow \begin{cases}
			\rho_I,\vec v_I, & \vec\xi\in V_I, \\
			\rho_Q,\vec v_Q, & \vec\xi\in V_Q
		\end{cases}\qquad\text{as $s\rightarrow\infty$,} \myeqlabel{eq:zero-cont-rhov}
	\end{alignat}
	i.e.\ their limits on rays to infinity are exactly as for the initial data in Figure \ref{fig:rrefini}. 
	This means the limit approaches the initial data as $t\downarrow 0$.

	\myeqref{eq:zero-weak}, \myeqref{eq:wall-eps}, \myeqref{eq:zero-cont} and \myeqref{eq:zero-cont-rhov} show that 
	$\phi(t,\vec x):=\psi^{(0)}(t^{-1}\vec x)$ defines
	a solution of \myeqref{eq:prob1}, \myeqref{eq:prob2}, \myeqref{eq:prob3} and \myeqref{eq:prob4}. 

	By taking $\overline\gamma\uparrow\infty$, $\underline\eta^*_C\downarrow 0$ 
	and $\underline\xi_{AB}\downarrow\xi_{EB}$,
	we obtain a solution for \emph{every} 
	$\gamma\in[1,\infty)$, $\eta_C^*\in[\eta^0_C,0)$ and $\xi_{AB}\in(\xi_EB,v^x_R]$.
	(in the cases $\gamma>1$ and $\eta_C^*=\eta^0_C$, we may use that $\overline\eta_C^*$ approaches $\eta^0_C$ as $\epsilon\downarrow 0$).
		
	As mentioned (Remark \myref{rem:Lambda-full}), this exhausts all cases covered by the 
	conditions of Theorem \myref{th:elling-rrefl}. The proof is therefore complete.
\end{proof}

\begin{remark}
	\mylabel{rem:structure}%
        In addition to mere existence we obtain some structural information in the proof:
	\begin{enumerate}
	\item The solution has the structure shown in Figure \myref{fig:rrsmr} left, with 
		pseudo-Mach number $L>1$ in the $I,R,Q$ regions, $L<1$ in the elliptic 
		region $\Omega$.
	\item The solution has constant density and velocity in each of the $I,R,Q$ regions. 
	\item The solution is analytic everywhere except perhaps at $\overline P^{*(0)}$ and in $\vec\xi_{AB}$ and, of course, the shocks.
	\item The curved shock is analytic away from $\hat A$ and $\overline P^{*(0)}$ and Lipschitz overall.
	\item Density and velocity are bounded. 
	\end{enumerate}
	It is expected that density and velocity are at least continuous. 
	However, the methods developed in \cite{elling-liu-pmeyer-arxiv} yield boundedness everywhere, but 
	continuity
	only away from $\overline P^*$.
	Note that $\overline P^*$ can \emph{not} be a classical shock with smooth data on each side,  
	because the one-sided limit of $L$ on the hyperbolic side $R$ of $P^*$ is $=1$ everywhere ($>1$ is needed for positive shock strength).

	Some additional structural information: 
	\begin{enumerate}
	\item The possible (downstream) normals of the curved shock are between $\vec n_R$ and $\vec t_A$ (counterclockwise). 
	\item The shocks are admissible and do not vanish anywhere.
	\item In the elliptic region, $v^x<v^x_I$ and $v^y\geq 0$ (in Figure \ref{fig:rrsmr} left coordinates).
	\item In the elliptic region, the density $\rho$ is greater than $\rho_I$. 
	\end{enumerate}
	Additional information can be obtained from the inequalities in Definition \myref{def:fusp}.
\end{remark}

\bibliographystyle{amsalpha}

\newcommand{\etalchar}[1]{$^{#1}$}
\providecommand{\bysame}{\leavevmode\hbox to3em{\hrulefill}\thinspace}
\providecommand{\MR}{\relax\ifhmode\unskip\space\fi MR }
\providecommand{\MRhref}[2]{%
  \href{http://www.ams.org/mathscinet-getitem?mr=#1}{#2}
}
\providecommand{\href}[2]{#2}

\end{document}